\documentclass[9pt]{entcs}

\usepackage{entcsmacro}
\usepackage{amsmath,amsfonts,amssymb}
\usepackage{graphicx}
\usepackage{diagrams,tikz}
\usepackage{stmaryrd}
\SetSymbolFont{stmry}{bold}{U}{stmry}{m}{n}
\usepackage{cite}

\newcommand{\cat}[1]{\ensuremath{\mathbf{#1}}}

\newarrow{Functor}====>

\newarrow{Map}---->
\newarrow{Embedding}C-+->
\newarrow{Mono}C--->
\newarrow{Epi}----{>>}
\newarrow{Open}o---{>>}
\newarrow{Closed}{triangle}--->
\newarrow{Iso}=====

\newarrow{AllMap}{.}{.}{.}{.}>
\newarrow{AllEmbedding}C{.}{+}{.}>
\newarrow{AllMono}C{.}{.}{.}>
\newarrow{AllEpi}{.}{.}{.}{.}{>>}
\newarrow{AllOpen}o{.}{.}{.}{>>}
\newarrow{AllClosed}{triangle}{.}{.}{.}>

\newarrow{ExistMap}{dash}{dash}{}{dash}>
\newarrow{ExistEmbedding}C{dash}{+}{dash}>
\newarrow{ExistMono}C{dash}{}{dash}>
\newarrow{ExistEpi}{dash}{dash}{}{dash}{>>}
\newarrow{ExistOpen}o{dash}{}{dash}{>>}

\newcommand{\id}[1]{\ensuremath{\text{id}_{#1}}}

\newcommand{\set}{\cat{Set}}

\newcommand{\op}[1]{{#1}^{\text{op}}}

\newcommand{\mcal}[1]{\mathcal {#1}}
\newcommand{\step}[1]{\xrightarrow {#1}}
\newcommand{\eseq}{\varepsilon}
\newcommand{\act}{A}
\newcommand{\tact}{\act_{\tau}}

\newcommand{\dom}[1]{\text{dom}(#1)}

\renewcommand{\phi}{\varphi}
\newcommand{\chaos}{\Omega}
\newcommand{\fairness}[1]{\text{Fair}_{#1}}

\newcommand{\fseq}[1]{{#1}^\star}
\newcommand{\iseq}[1]{{#1}^\omega}
\newcommand{\aseq}[1]{{#1}^\infty}
\newcommand{\history}[1]{\downarrow #1}
\newcommand{\tuple}[1]{(#1)}
\newcommand{\asheaf}{\mathfrak a}
\newcommand{\isheaf}{\mathfrak i}
\newcommand\inv[1]{#1\raisebox{1.15ex}{$\scriptscriptstyle-\!1$}}
\newcommand{\nat}{\mathbb{N}}
\newcommand{\semantics}[1]{\llbracket #1 \rrbracket}
\newcommand{\lts}{\cat{LTS}}

\newcommand{\exec}[1]{\text{Exec}(#1)}
\newcommand{\mexec}[1]{\text{MExec}(#1)}
\newcommand{\fexec}[1]{\text{FExec}(#1)}

\newcommand{\sheaves}[1]{\cat{Sh}(#1)}
\newcommand{\presheaves}[1]{\cat{PSh}(#1)}

\newcommand{\mpast}{\text{mpast}}
\newcommand{\ftact}{\fseq{\act}_{\tau}}

\newcommand{\A}{\mathcal A}

\newcommand{\obs}{\mathcal O}

\newcommand{\elements}[1]{\mathbb{E}(#1)}
\newcommand{\yoneda}[1]{\mathbb{Y}_{#1}}
\newcommand{\colimit}[1]{\varinjlim_{#1}}

\newcommand{\last}[1]{\text{last}(#1)}
\newcommand{\timeaxis}{\cat T}
\newcommand{\cod}[1]{\text{cod}(#1)}
\newcommand{\fts}{\cat{FTS}}
\newcommand{\ftbact}{\act_{\bar\tau}^\star}
\newcommand{\ftbtact}{\act_{\tau\bar\tau}^\star}

\makeatletter
\providecommand*{\twoheadrightarrowfill@}{%
  \arrowfill@\relbar\relbar\twoheadrightarrow
}
\providecommand*{\xtwoheadrightarrow}[2][]{%
  \ext@arrow 0359\twoheadrightarrowfill@{#1}{#2}%
}
\newcommand{\steps}[1]{\xtwoheadrightarrow{#1}{}}
\makeatother

\newtheorem{notation}[thm]{Notation}

\usepackage{tikz}
\usetikzlibrary{shapes,arrows}
\usetikzlibrary{calc}

\makeatletter
\newcommand{\eqnum}{\refstepcounter{equation}\textup{\tagform@{\theequation}}}
\makeatother

\def\lastname{Beohar, K\"upper}
\begin{document}
\begin{frontmatter}
%\raggedbottom
%
\title{Bisimulation maps in presheaf categories}
%Add this acknowledgement after acceptance.
%\thanks{Partially supported by DFG project BEMEGA.}
%
%\titlerunning{Abbreviated paper title}
% If the paper title is too long for the running head, you can set
% an abbreviated paper title here
%

\author{Harsh Beohar\thanksref{myemail}}
  \address{Universit\"at Duisburg-Essen\\ Duisburg, Germany} \author{Sebastian K\"upper\thanksref{coemail}}
  \address{FernUniversit\"at in Hagen\\Hagen, Germany}  \thanks[myemail]{Email:
    \href{mailto:harsh.beohar@uni-due.de} {\texttt{\normalshape
        harsh.beohar@uni-due.de}}} \thanks[coemail]{Email:
    \href{mailto:sebastian.kuepper@feu.de} {\texttt{\normalshape
        sebastian.kuepper@feu.de}}}

%
%\maketitle              % typeset the header of the contribution
%
\begin{abstract}
%The abstract should briefly summarize the contents of the paper in 150--250 words.
The category of presheaves on a (small) category is a suitable semantic universe to study behaviour of various dynamical systems. In particular, presheaves can be used to record the executions of a system and their morphisms correspond to simulation maps for various kinds of state-based systems. In this paper, we introduce a notion of bisimulation maps between presheaves (or executions) to capture well known behavioural equivalences in an abstract way. We demonstrate the versatility of this framework by working out the characterisations for standard bisimulation, $\forall$-fair bisimulation, and branching bisimulation.

\end{abstract}
\end{frontmatter}
\section{Introduction}
%1)       The technical problem to be solved with a justification of its importance.

The importance of formal semantics should not be underestimated, especially when aimning to design \emph{reliable} dynamical systems in heterogeneous environments. Therefore, a variety of state based modelling frameworks at different levels of abstraction have been proposed; to quote Goguen \cite{Goguen92sheafsemantics}: \emph{one person's syntax is another person's semantics}.
Diversity in %tools and
algorithms can be desirable; however, as %noted semanticists have
argued in \cite{Goguen92sheafsemantics,abramsky:critic-2006,presheaves-as-transitionsys:1997}, the proliferation of semantic theories indicates our scattered understanding of concurrent systems.
Thus, we seek a framework that provides semantic structure describing the behaviour of a dynamical system and its refinement independently of syntax.

%2)       An account of related and prior work, and an explanation why this has not solved the problem.
This goal is shared to an extent by the theory of coalgebras \cite{Rut03:universal}. %In Section~\ref{sec:related}, we will compare our proposal with coalgebras and other semantic frameworks.
In \cite{bk2017}, we abandoned state-based modelling in favour of describing behaviour as the set of executions (inspired by \cite{Cuijpers:2013:DCM,control-theory-book}) because the branching structure of a state in the presence of invisible actions is described by the set of executions (not states).
%one based on executions .
%Our solution was not entirely abstract in nature and we believe that executions are hard to capture coalgebraically due to the inherent step-based nature of a coalgebra.
This situation further escalates when one is interested in
infinite executions (e.g. fairness properties \cite{hennessy:futureperfect,Kwiatkowska1989:fairness-survey}) or dense executions which are omnipresent in hybrid systems (e.g. \cite{Cuijpers:lost-in-translation}). The point is not that executions are inexpressible in a coalgebra, but rather that we need a semantic framework where they are treated as first-class citizens just as states are in a coalgebraic framework. Thus, our hypothesis is that behaviour of a system is given solely by its executions. %In a sense, behaviour of a system can be solely given by defining its executions. %, something the step-based coalgebraic modelling technique is not very well-suited to.

%3)       The research hypothesis or claim.
We anticipate presheaves to be the ``right'' semantic structure to study executions without fixing a kind of dynamical system. Note that we are \emph{not} the first in proposing presheaves as the mathematical universe to studying behaviour. Winskel and his colleagues \cite{presheaves-as-transitionsys:1997,cattani_winskel_2005,hildebrandt:fairness,jnw96:bisimopenmaps,Fiore:wbisim_open-maps}
have already employed presheaves (among other things) by giving a denotational semantics of process-algebraic terms supported by characterisations of strong bisimulation and weak bisimulation relations using open maps in the context of transition systems. For a more modern treatment, Hirschowitz and his colleagues \cite{Hirschowitz:LICS19,Hirschowitz:sheaf-pi,hirschowitz:inncocent} advocated game semantics using (pre)sheaves. %and event structures.

The novelty of our work lies in refining the notion of open maps (which we christened \emph{bisimulation maps}) in a presheaf category and using it to characterise $\forall$-fair bisimulation \cite{Kupferman2003:fair_equiv_rel,Henzinger:2002} and branching bisimulation \cite{bbisim:1996} relations. In addition, the prospect of having to specify notions of time and observation (which was absent in the earlier works on presheaf semantics) leads to a clearer modelling, so explicitly highlighting  these two dimensions of system modelling is at the core of our contribution (cf. Section~\ref{sec:prelim}). This distinction was in turn essential to capture branching bisimulation in the presence of invisible actions $\tau$.

Interestingly, unlike open maps, bisimulation maps are always retracts in the category of presheaves (in turn, they are surjective at the level of executions). As a slogan, presheaf maps are refinement maps, while bisimulation maps (which are special presheaf maps) are complete refinement maps (Section~\ref{sec:bisimmaps}). By moving to a finer notion, we are still able to capture functional bisimulations without fairness. However, in the context of fairness, we can show (Theorem~\ref{thm:existsfairisfairbisimfct}) that the behavioural equivalence induced by a bisimulation map coincides with $\forall$-fair bisimulation relation. Note that our $\forall$-fair bisimulations are equivalence relations by definition in contrast to the existing definition  \cite{Kupferman2003:fair_equiv_rel,Henzinger:2002} (see the dicussion after Theorem~\ref{thm:existsfairisfairbisimfct} on Page~\pageref{thm:existsfairisfairbisimfct}). This is an improvement with respect to the previous characterisation of $\exists$-fair bisimulation \cite{Kupferman2003:fair_equiv_rel,Henzinger:2002} (called \emph{extended bisimulation} in \cite{hennessy:futureperfect}) obtained by Hildebrandt \cite{hildebrandt:fairness} using open maps, since any $\forall$-fair bisimulation relation is strictly finer than an $\exists$-fair bisimulation relation \cite{Henzinger:2002} and our correspondence does not impose any restrictions on the fairness predicates. These restrictions, originally from \cite{hennessy:futureperfect}, asserted that fairness predicates on infinite executions are closed under the removal and the addition of finite prefixes.

%In retrospect, \citeauthor{Goguen92sheafsemantics} \cite{Goguen92sheafsemantics} was first in applying sheaf theory to give semantics of object-based concurrent programming languages in general. However, the novelty of our work lies in identifying executions as germs of a sheaf and capturing (bi)simulation functions using the vocabulary of sheaf theory \cite{sheafbook}.

Another practical aspect of the theory of presheaves is that it guides us in finding the right semantic categories once a notion of time and observation is fixed. %(cf. Subsection~\ref{subsec:guide}).
%before even specifying a dynamical system (Section~\ref{sec:prelim}). As a consequence,
Moreover, we can apply concepts (like, e.g., essential geometric morphism \cite{sheafbook}) that transform a dynamical system from one observation space $\obs$ to another space $\obs'$. %This has applications in characterising branching bisimulation to deal with invisible transitions (often denoted by $\tau$).
This way we can transform (see Section~\ref{sec:invsible}) a presheaf of executions (induced by a given transition system) into a presheaf of minimal executions (i.e., executions in which trailing $\tau$-transitions are chopped off). This property is specific to branching bisimulation, which may be the reason why this construction was not discussed in \cite{Fiore:wbisim_open-maps} (their objective was to capture weak bisimulation).

%Yet another use of (pre)sheaves, which we have not explored in this paper, could be in identifying models (or sound axiomatisations) of process theories. As an example, consider the theory $\mathbf{BPA}_{0,1}$ \cite{tcp-book} similar to the first order theory of rings, where the axioms of additive inverse and left distributivity are absent and replaced by the idempotence of addition. For such theories, there are general results based on classifying toposes \cite{sheafbook} that help in finding universal models (similar to the concept of initial algebra in universal algebra). We are not aware to what extent these results are applicable in the context of process theories. However, before exploring such avenues, we should be confident that well known behavioural equivalences can be captured in a topos because after all the axioms of process theories are designed to be valid only up to the chosen equivalence.

\paragraph{Organisation of the paper.}
In Section~\ref{sec:prelim}, we introduce our mathematical framework to model behaviour of a dynamical system with a special focus on the aspects of time and observation. Then, we introduce the notion of bisimulation maps in presheaves on an arbitrary (small) category in Section~\ref{sec:bisimmaps}. Turning our attention towards the first major example, we characterise $\forall$-fair bisimulation relations in Section~\ref{sec:fairness}. The case of invisible actions in Section~\ref{sec:invsible} is based on a change of observation space. We first outline an obvious (but ultimately failed) attempt to capture branching bisimulation, before giving the correct (yet intuitive) construction that characterises branching bisimulation. %All the full proofs pertaining to each section can be found in the clearly marked appendix.

%\section{Preliminaries}\label{sec:prelim}
\section{Our universe of discourse}\label{sec:prelim}
The objective of this section is to describe our semantic framework in which one can model behaviour of a dynamical system. By \emph{behaviour} of a dynamical system, we understand some \emph{phenomena that evolve over time}. Our aim is to formalise this intuition. We begin by modelling time as a small category $\timeaxis$, whose objects are points in time and arrows describe passing of time.

\begin{notation}
  An object $C$ (an arrow $f$) of a category $\cat C$ will be denoted by the predicate $C\in \cat C$ ($f\in\cat C$). Moreover, the codomain and domain of an arrow $f\in\cat C$ are denoted as $\cod f$ and $\dom f$, respectively.
\end{notation}

Invariably, dynamical systems come with a notion of observation.
For instance, a letter from a fixed alphabet may denote the assignment of model variables in a computer program/controller. We assume that a system under study has a display unit together with the existence of a hypothetical `observer' $\obs$ who is watching/measuring behaviour of the system using this display unit over time. In addition, our observer $\obs$ can remember its observations over time, i.e., earlier observations can be deduced from the later observations. Mathematically, this amounts to saying that $\obs$ is a contravariant functor $\timeaxis \rTo \set$.

\begin{proposition}
  Let $\cat C$ be a small category. Then, the collection of functors of type $\op{\cat C} \rTo \set$ (i.e., presheaves on a category $\cat C$) and natural transformations between them form a category $\presheaves{\cat C}$.
\end{proposition}

\begin{notation}
  Given a presheaf $F\in\presheaves{\cat C}$, we follow \cite{sheafbook} in writing $x\cdot f$ to denote the restriction of $x\in FC$ along $C'\rTo^f C$, i.e., $Ff(x)=x\cdot f$. In case $\cat C$ is a poset (viewed as a category) $\cat C$, we write $x\cdot C'$ to denote the restriction of $x\in FC$ along $C'\preceq C$. %i.e., $F((C'\preceq C)_{\text{op}})(x)=x\cdot C'$.
  Note that we use calligraphic letters for specific presheaves, whereas arbitrary ones are denoted by capital letters as above.
\end{notation}

\begin{example}\label{ex:observation-lts}
  In this example, we fix the notion of time $\timeaxis$ and observation $\obs$ associated with a (labelled) transition system.
  For time $\timeaxis$ we take the set of natural numbers $\nat$ viewed as a category (arrows are the less-than-equal-to relations). For the given alphabet $\act$, we now define a presheaf $\A\in\presheaves{\nat}$:
  \[
  \A(n) =
  \{\sigma\in\fseq\act \mid |\sigma| = n\} \qquad
    \text{(for every $n\in\nat$),}
  \]
  together with the action on $\A$ given by $\sigma\cdot n =\sigma|_n$ (for every $\sigma\in \A(n')$ and $n\leq n'$).
  %\[\A(n\leq n')(\sigma) = \sigma|_{n}, \qquad \text{(for $\sigma\in \A(n')$)}.\]
  In other words, $\A(n)$ is the set of those finite words $\sigma\in \fseq\act$ whose length is $n$ (denoted by $|\sigma|=n$), while the action $\_\cdot n$ simply maps a word $\sigma$ of length $n'$ to its unique prefix of length $n$ (denoted by $\sigma|_n$). Note that $\A(0)$ is a singleton set containing the empty word which we denote by $\eseq$.
\end{example}
\begin{remark}\label{remark:sheaves}
  In modelling some dynamical systems, like, e.g., those arising from control theory \cite{control-theory-book}, $\obs$ may have even more structure in that global observations can be constructed by gluing  the local observations (smaller neighbourhoods). In such situations, the category of sheaves $\sheaves{\cat C,J}$ equipped with a Grothendieck topology $J$ on $\cat C$ is more suitable (cf. \cite{spivak_hybrid-systems}) for semantic purposes. Moreover, sheaves equipped with discrete Grothendieck topology are exactly presheaves (cf.\thinspace\cite{sheafbook}), so our mathematical universe is actually the category of sheaves (rather than presheaves). But due to the discrete nature of dynamical systems considered in this paper, we restrict ourselves to presheaves. Nevertheless, we will state our definitions so that they are applicable on sheaves (see, e.g., Remark~\ref{remark:sheaves-bisim}).
\end{remark}
Once a notion of time $\timeaxis$ and an observation $\obs\in\presheaves{\timeaxis}$ is fixed, then a system essentially describes the \emph{runs} (also known as \emph{trajectories} or \emph{executions}) of the system and the observation associated with each run.
%\begin{itemize}
%    \item what are the \emph{runs} (also known as \emph{trajectories} or \emph{executions}) of the system?
%    \item what is the observation associated with each run?
%\end{itemize}
To answer both, we envisage that a dynamical system is nothing but an object in the slice category $\presheaves{\cat C}/\obs$. In other words, a dynamical system corresponds to a presheaf $F$ modelling the runs of the system and a natural transformation $F \rTo^\alpha \obs$ modelling the observation associated with each run of the system. More importantly, a \emph{system homomorphism} $\phi$ between two systems $(F,\alpha)$ and $(G,\beta)$, denoted $(F,\alpha) \rTo^\phi (G,\beta)$, is a natural transformation $F \rTo^\phi G$ preserving the observations, i.e., $\beta \circ \phi = \alpha$. Intuitively, a system homomorphism $(F,\alpha) \rTo^\phi (G,\beta)$ says that the system $(F,\alpha)$ is a \emph{refinement} of $(G,\beta)$ (i.e., every observable behaviour of $F$ is also part of the observable behaviour of $G$).

\subsection{Refining our framework by unifying time and observation}\label{subsec:guide}
Although the slice category $\presheaves{\timeaxis}/\obs$ is close to our system theoretic intuition, its presentation can be further simplified.  % using the category of elements. %Given a presheaf $F\in\presheaves{\cat C}$, then
Recall the \emph{category of elements} of a presheaf $F\in\presheaves{\cat C}$, denoted $\mathbb E_{\cat C}(F)$ (we drop the subscript $\cat C$ whenever clear from the context), has as objects the tuples $(x,C)$ with $C\in\cat C,x\in FC$ and as arrows $(x,C) \rTo^f (x',C')$ the morphism $C \rTo^f C'\in\cat C$ such that $x'\cdot f=x$.

\begin{theorem}[\!\!\protect{\cite[Exercise~III.8(a)]{sheafbook}}]\label{thm:elements-equivalence}
  For a presheaf $F$ over a small category $\cat C$, there is an equivalence of categories
  $\presheaves{\cat C}/F \cong \presheaves{\elements{F}}$.
\end{theorem}
\noindent
Note that a similar result also holds in the setting of sheaves (cf.\thinspace \cite[Exercise~III.8(b)]{sheafbook}).
%\footnote{A similar result also holds in the setting of sheaves (cf.\thinspace \cite[Exercise~III.8(b)]{sheafbook}).}

In other words, time can be made inherent with observation and, thus, we can work in a simpler setting without worrying about the bookkeeping associated with slice categories. To see this, recall Example~\ref{ex:observation-lts} and the poset of finite words $\fseq\act$ (a.k.a.\thinspace free monoid) generated by a set $\act$, which is ordered by the prefix relation $\preceq\ \subseteq\fseq\act\times\fseq\act$. Notice that the categories $\elements{\A}$ and $\fseq\act$ are isomorphic: since the length of a word is redundant in the objects of $\elements{\A}$ dropping the length results in the elements of $\fseq\act$. Thus, we obtain
\begin{corollary}
There is an equivalence of categories $\presheaves{\nat}/\A \cong \presheaves{\fseq\act}$.
\end{corollary}

As a result, the category of presheaves on $\fseq\act$ can serve as the semantic universe to study behaviour of a transition system (cf. Example~\ref{ex:semantics-lts}). More generally, by giving the semantics to a `syntactic' category of a computational model $\cat M$, we mean identifying the notion of time $\cat T$ and observation $\obs\in\presheaves{\cat T}$ together with a faithful functor $\cat M \rTo^{\semantics{\_}} \presheaves{\mathbb E_{\cat T}(\obs)}$, called the \emph{semantics} functor. By interpreting an arrow $M \rTo^f M'$ in $\cat M$ as \emph{$M$ is an implementation of $M'$ witnessed by $f$}, then faithfulness of $\semantics{\_}$ asserts: if an implementation is witnessed by two semantically same morphisms $\semantics{f}=\semantics{g}$, then $f=g$ must be the same syntactically. %these morphisms must be the same syntactically, i.e., $f=g$.

%TODO: explain that (unlike previous work of presheaves and coalgebra) we do not care whether the semantic universe is in one-to-one correspondence with the category of models.

\begin{example}\label{ex:semantics-lts}
  Consider a transition system $\tuple{X,\act,\rightarrow}$ where $X$ is the set of states, $\act$ is the set of actions, and $\rightarrow \subseteq X \times \act \times X$ is the transition relation\footnote{Transition systems without initial states are standard in process algebraic literature (see \cite{process_algebra_reference}).}. Then the collection of transition systems together with simulation functions form a category denoted $\lts$. Note that a \emph{simulation function} %$f$ between $\tuple{X,\act,\rightarrow}$ and $\tuple{Y,\act,\rightarrow}$
  is a function $X \rTo^f Y$ satisfying:
  \begin{equation}\label{eq:trans-preserve}
    \forall_{x,x'\in X,a\in\act}\ x\step a x' \implies f(x) \step a f(x').
  \end{equation}
  As usual, we write $x \step a x'$ to denote $(x,a,x')\in\rightarrow$. Let $\history{\sigma}=\{\sigma'\in\fseq\act\mid \sigma' \preceq \sigma \}$ be the prefixes of $\sigma$.
  Next, define a presheaf $\semantics{X}\in\presheaves{\fseq\act}$ which records all the executions whose trace is $\sigma$ at $\semantics{X}(\sigma)$:
  \begin{align*}
    \semantics{X}(\sigma) =&\ \left\{\history{\sigma} \rTo^p X \mid \forall_{\sigma',a}\ \big(\sigma'a\preceq\sigma \implies p(\sigma') \step a p(\sigma' a)\big) \right\},\\
    p\cdot \sigma' =&\ p_{\history{\sigma'}} \quad \text{\big(for any $\sigma'\preceq \sigma $ and $p\in\semantics{X}(\sigma)$\big)}.
  \end{align*}
  %and the action is given by restricting the domain of an execution, i.e., $p\cdot \sigma' = p|_{\history{\sigma'}}$ (for any $\sigma'\preceq \sigma $ and $p\in\semantics{X}(\sigma)$)..
  Moreover, any function $X \rTo^f Y$ induces a family of maps $\semantics{f}_\sigma (p) = f\circ p$ (for each $\sigma\in\fseq\act$).
\end{example}
Thus, we obtain the following result; wherein, the result that presheaf maps encode simulation maps is well known from the early work of Joyal {et al.} in \cite{jnw96:bisimopenmaps}.
\begin{proposition}\label{prop:lts-semantics}
  The map $\semantics{\_}$ defined in Example~\ref{ex:semantics-lts} is a faithful functor. Moreover,
  \begin{itemize}
    \item if $\semantics{f}$ is a presheaf map for any $f\in\lts$, then $f$ is a simulation function.
    \item if a function $f$ between the underlying state spaces induces a presheaf map $\semantics f$, then $f$ is a simulation map.
  \end{itemize}
\end{proposition}
As a result, the category of transition systems $\lts$ has presheaf semantics on $\fseq\act$. In the subsequent sections, we will show the applicability of our semantic framework by giving presheaf semantics to different computational models; namely, transition systems with fairness predicates and invisible transitions.
%\begin{remark}
%  In our earlier attempts, we actually associated a sheaf over the topological space $\fseq\act$ with every given transition system, where $\fseq\act$ was endowed with the Alexandrov topology induced by the prefix order $\preceq$. But we learnt from Dr.\thinspace Alex Simpson (personal communication) that the category of sheaves over the space $\fseq\act$ is equivalent to the category of presheaves $\presheaves{\fseq\act}$, which is easier to work with.
%\end{remark}
\section{Bisimulation maps: towards complete refinement}\label{sec:bisimmaps}

Consider a category of computational models $\cat M$ together with its semantics over a presheaf category $\presheaves {\cat C}$, i.e., $\cat M \rTo^{\semantics{\_}} \presheaves{\cat C}$. As mentioned earlier, an arrow between any two images of the semantics functor as a refinement map from the modelling point of view. In particular, we interpret $\semantics{M} \rTo^{\semantics{f}} \semantics{M'}$ (induced by an arrow $M \rTo^f M'\in\cat M$) as the information that $\semantics{M}$ is an implementation of $\semantics{M'}$ witnessed by the refinement map $\semantics{f}$. Note that this is a straightforward generalisation of Proposition~\ref{prop:lts-semantics} since simulation maps encode the refinement of behaviour in the case of labelled transition systems. %\todo{Added footnote due to remark (2) of Reviewer (1).}
%It should be clear that presheaf morphisms model refinement of behaviour in the sense that $\semantics{M}$ can be viewed as an implementation of $\semantics{M'}$ witnessed by $\semantics{f}$, whenever $M\rTo^f M'\in\cat M$\footnote{The intuition behind the usage of the term 'implementation' that an implementation should show all modelled behaviour, but but may be more specific through its implementation details than the model, thus potentialls showing additional behaviour.} Now we restrict to those presheaf morphisms that represent `complete' refinement of behaviour, i.e., every extension of behaviour observable in $\semantics {M'}$ can be reflected as an extension observable in $\semantics M$ through the witnesser $\semantics f$.
Therefore, to study bisimulation maps at the level of presheaves (executions), we will restrict ourselves to those presheaf morphisms that represent `complete' refinement of behaviour in the following sense. Theorem~\ref{thm:bisim-epi} expresses this in a more formal manner. %every extension of behaviour observable in $\semantics {M'}$ can be reflected as an extension observable in $\semantics M$ through the witnesser $\semantics f$.
\begin{itemize}
  \item Every observable behaviour in $\semantics{M'}$ is also observable in $\semantics{M}$ (thus, $\semantics{M}$ is a refinement of $\semantics{M'}$).
  \item Moreover, every observable behaviour in $\semantics{M}$ can be retracted onto an observable behaviour in $\semantics {M'}$.
\end{itemize}

One possibility is to use the open maps of Joyal et al. \cite{jnw96:bisimopenmaps}, which gave a unified definition of functional bisimulations over the range of computational models. In particular, when invoked in a presheaf category, open maps correspond to natural transformations whose naturality squares are the weak-pullback squares in $\set$ (see \cite[Proposition~2.3]{cattani_winskel_2005}). Nevertheless, an open map falls short in capturing the complete refinement point of view since an arbitrary open map may not even be a surjective map at the level of executions; i.e., our implementation may not even implement or cover all the behaviour present in the specification. %; thus fail to capture the above complete refinement point of view.

In \cite{openmaps-concrete:2015}, open maps \cite{jnw96:bisimopenmaps} were refined to embedding-open maps in the setting of concrete categories. An important difference to the classical definition of open maps is that the parametric notion of path extensions can be replaced by embeddings (see \cite[Definition~8.6(2)]{book:con-cat} for a formal definition).
%\begin{wrapfigure}[6]{r}{4cm}
%\vspace{-1.4cm}
%
%\end{wrapfigure}
Moreover, embedding-open maps are always retracts under some mild restrictions on concrete categories (cf. \cite{openmaps-concrete:2015}). Now, if a category is concrete over itself, then an embedding corresponds to simply a monomorphism. Consequently, we propose the following definition of an embedding-open map, which we simply call a \emph{bisimulation} map.
%In this section, we will apply this definition on the semantics functor for a chosen category of computational models. Recall that a category $\cat C$ is \emph{concrete} \cite{book:con-cat} over a category $\cat D$ if there is a faithful functor $\cat C \rTo^F \cat D$. %In such a situation, we simply write $\cat C \rTo^F \cat D$ to denote a concrete category.
%
%\begin{definition}\label{def:embedding}
%  In a concrete category $\cat C \rTo^F \cat D$, a map $C \rTo^f C' \in \cat C$ is an \emph{embedding}, denoted $C \rEmbedding^f C'$, if and only if the following two conditions hold.
%  \begin{itemize}
%    \item (underlying monos) $\forall_{g,h\in \cat D}\ Ff\circ g =Ff \circ h \implies g=h$;
%    \item (initial) $\forall_{g\in\cat D,h\in\cat C}\ Fh= Ff\circ g \implies \exists_{g'\in \cat C}\ Fg'=g$.
%  \end{itemize}
%\end{definition}
\begin{definition}\label{def:bisim-presheaf}
  A map $F \rTo^f G \in \presheaves{\cat C}$ is a  \emph{bisimulation} if, and only if, for every commutative square depicted in \eqref{eq:bisim-presheaf} with a mono $P \rMono^g Q$ and maps $m,n$ in $\presheaves {\cat C}$, there is a map $Q \rTo^k F\in\presheaves{\cat C}$ such that the two triangles commute, i.e., $k\circ g=m$ and $f\circ k=n$.
  \begin{diagram}[LaTeXeqno]\label{eq:bisim-presheaf}
    Q & \rTo^{n} & G \\
    \uMono^{g} & \rdExistMap^k & \uTo_{f} \\
    P & \rTo_{m} & F
  \end{diagram}
\end{definition}
Intuitively, \eqref{eq:bisim-presheaf} states that every extension of behaviour observable in $G$ can be reflected as an extension observable in $F$ through the witnesser $f$.
%\begin{wrapfigure}[10]{r}[]{4cm}
%  \begin{equation}\label{eq:bisim-presheaf}
%    \begin{diagram}
%    Q & \rTo^{n} & G \\
%    \uMono^{g} & \rdExistMap^k & \uTo_{f} \\
%    P & \rTo_{m} & F
%  \end{diagram}
%  \end{equation}
%  \end{wrapfigure}
\begin{remark}\label{remark:sheaves-bisim}
  It is interesting to note that a similar definition for sheaves can be derived from the definition of embedding-open maps as given in \cite{openmaps-concrete:2015}. First, note that $\sheaves{\cat C,J}$ is concrete over $\presheaves{\cat C}$ due to the forgetful functor $\isheaf$, which is fully faithful. Moreover, embeddings in this concrete category are actually monomorphisms. This is because any mono is a regular mono in $\sheaves{\cat C, J}$ and the faithful functor $\isheaf$ preserves regular monos (since $\isheaf$ is right adjoint to the associated sheaf functor $\asheaf$ \cite{sheafbook}). Lastly, every regular mono is an embedding whenever the faithful functor preserves regular monos \cite[Proposition~8.7.3]{book:con-cat}. Thus, we have the exact same definition of bisimulation maps in $\sheaves{\cat C,J}$.
\end{remark}
%\begin{theorem}[\cite{openmaps-concrete:2015}]\label{thm:emb-open-epis}
%  In a concrete category $\cat C \rTo^F \cat D$, if $\cat C$ has an initial object and $F$ preserves it, then every embedding-open map in $\cat C$ is an epimorphism.
%\end{theorem}
\begin{theorem}\label{thm:bisim-epi}
  Every bisimulation map in a (pre)sheaf category is a retract.%; thus, every bisimulation map is an epimorphism.
\end{theorem}
We end this section by capturing functional bisimulations in terms of bisimulation maps whose proof can be extracted from the proof of Theorem~\ref{thm:fair-bisimulation-map}. Note that a similar theorem was proven earlier in the seminal paper \cite{jnw96:bisimopenmaps} for functional bisimulation; however, the difference is that we use bisimulation maps (not open maps) in our characterisation.
\begin{theorem}
Given a simulation function $X \rTo^f Y$, then $\semantics{f}$ is a bisimulation map in $\presheaves{\fseq\act}$ if, and only if, the function $f$ is a surjection satisfying: %and reflects the transitions of $Y$, i.e.,
\begin{equation}\label{eq:bisim-reflect}
  \forall_{x\in X,y\in Y}\ \big( f(x) \step a y \implies \exists_{x'\in X}\ (x\step a x' \land f(x')=y) \big).
\end{equation}
%\begin{enumerate}
%  \item The function $f$ is a surjection and reflects the transitions of $Y$, i.e.,
%      \[\forall_{x,y}\ fx \step a y \implies \exists_{x'}\ x\step a x' \land fx'=y\]
%  \item The map $\semantics{X} \rTo^{\semantics{f}} \semantics{Y}$ is a bisimulation map in $\presheaves{\fseq\act}$.
%\end{enumerate}
\end{theorem}

\section{The case of fairness}\label{sec:fairness}

When considering infinite, rather than finite, behaviour of systems, it is common to take fairness  into account. There are various notions of fairness \cite{Henzinger:2002} but the general idea for fair (bi)simulation is to demand that fair executions are matched by fair executions, in addition to classical (bi)simulation properties. In this section, we give a presheaf semantics to fair transition systems and outline how $\forall$-fair bisimulation can be captured via bisimulation maps.

Let $\iseq\act=\{\sigma \mid \nat \rTo^\sigma \act\}$ be the set of infinite words generated from $\act$ and fix $\aseq\act=\fseq\act\cup\iseq\act$, which is ordered by the prefix relation $\preceq$. The object of study are \emph{fair transition system}s $\tuple{X,\act,\rightarrow, \fairness {X}}$, where $\tuple{X,\act,\rightarrow}$ is a transition system and $\fairness {X}$ is a fairness predicate on infinite executions. An \emph{infinite execution} is a function $\history{\sigma} \rTo^p X \cup \{\chaos\}$ whose domain is the history of an infinite word $\sigma\in\iseq\act$ such that
$p(\sigma) = \chaos$ and $\forall_{\sigma',a}\ (\sigma' a\prec \sigma \implies p(\sigma') \step a p(\sigma'a))$.
Here, $\Omega$ is the \emph{chaos} state which the system enters once it has executed an infinite execution. Let $\fexec{X}=\fairness{X}\cup\exec{X}$ be the set of all fair executions ordered by the prefix relation $\preceq$, i.e., $p\preceq p' \iff p'|_{\dom p} = p$ (for $p,p'\in\fexec{X}$). The set of (finite) executions $\exec{X}$ is defined as earlier, i.e., $\exec{X}=\{\history\sigma\rTo^p X\mid \sigma\in\fseq A\wedge\forall_{\sigma'a\in \history\sigma}\ p(\sigma')\step ap(\sigma'a)\}$. %\todo{Added def of exec X because of Reviewer (1) Remark (5)}% the above fair transition system.

\begin{definition}
\label{def:fairbisim}
 A chaos preserving extension $X\cup\{\Omega\} \rTo^{f_\Omega} Y \cup\{\Omega\}$ (i.e., $f_\Omega(x)=f(x)$ for $x\in X$ and $f_\Omega(\Omega)=\Omega$) of a function $X\rTo^f Y$ is a \emph{fair simulation} between $\tuple{X,\act,\rightarrow, \fairness {X}},\tuple{Y,\act,\rightarrow, \fairness {Y}}$ if, and only if, $f$ satisfies \eqref{eq:trans-preserve} and
  $
  \forall_{p\in\fairness{X}}\ f_\Omega\circ p\in \fairness {Y}
  $.
  Henceforth, we do not distinguish between $f_\Omega$ and $f$.
\end{definition}

Note that an infinite, but fair, execution can be seen as the limit of a monotonically increasing sequence of finite executions in $\fexec X$. Since this limiting sequence is part of behaviour, we should reflect it. Thus, we say a \emph{fair bisimulation} $X\cup \{\Omega\} \rTo^f Y \cup \{\Omega\}$ is a surjective fair simulation $f$ satisfying \eqref{eq:bisim-reflect} and the following condition for any increasing sequence of finite executions $(p_i)_{i\in\nat}$ in $X$:
%for any downward closed (w.r.t. $\preceq$) chain $U\subseteq \fexec{X},q\in\fexec Y$ we have
  \begin{equation}\label{eq:fbisim-reflect}
  \bigsqcup_{i\in \nat} f\circ p_i \approx f\circ \bigsqcup_{i\in\nat} p_i.
  %f(U) \subseteq\ \history{q} \implies \exists_{p\in \fexec{X}}\ (U \subseteq\ \history{p} \land f\circ p =q).
  %\forall_{p\in\fexec X,q\in\fexec Y}\ \big( f\circ p \preceq q  \implies \exists_{p'\in\fexec{X}}\ (p \preceq p' \land f\circ p'=q) \big).
  \end{equation}
%Note that the condition in \eqref{eq:bisim-reflect} and that $f$ is surjective are inbuilt in \eqref{eq:fbisim-reflect}.
Here, $\approx$ is the Kleene equality used to equate the  partially defined terms above.
  %\begin{enumerate}
%		\item $\forall_{x,x'\in X,a\in\act}\ \big( x \step a x' \implies fx \step a fx'\big)$, and
%		\item $\forall_{p\in\iexec{X}}\ \big(p\in \fairness {X}\implies f\circ p\in \fairness {Y}\big)$.
%	\end{enumerate}
%\end{definition}
%Now to define a fair bisimulation function, we need to understand what it means to reflect an extension of behaviour. It would be na\"{\i}ve to only require that any future $q$ of an execution $f\circ p$ in $Y$ with $p\in\exec{X}$ is due to a future of $p$ in $X$ (i.e., $p \preceq p' \land f\circ p'=q$, for some $p'$). This is because

%  \begin{enumerate}\setcounter{enumi}{2}
%  \item $\forall_{p,q}\ \big( (q\in\fairness{Y} \land f\circ p \preceq q ) \implies \exists_{p'\in\fairness{X}}\ (p \preceq p' \land f\circ p'=q) \big)$.
%  \end{enumerate}
%\end{definition}

To the best of our knowledge, the above notion of fair bisimulation is novel; however, below we will establish its connection with the literature after characterising it in terms of presheaf morphisms.
%The specific notion of fair bisimulation proposed here is -- to our best knowledge -- novel, but we will show that it actually coincides with the established notion of $\forall$-bisimulation.
So let $\fts$ be the category of fair transition systems and fair simulation between them. In addition,  %Next, we define the basic entities of our semantic framework to capture fair bisimulation.
\begin{itemize}
  \item Time: Since infinite executions are allowed in a fair transition system, we take $\timeaxis=\nat\cup\{\infty\}$ to be the category of natural numbers extended by a number representing infinity (i.e., $\forall_{n\in\nat}\ n\leq \infty$).
  \item Observation: Using the definition of $\A$ in Example~\ref{ex:observation-lts}, we define: $\obs(\infty)=\iseq\act$ and $\obs(n)=\A(n)$ (for $n\in\nat$).
\end{itemize}
Clearly, the categories $\elements{\obs}$ and $\aseq\act$  are isomorphic and by applying Theorem~\ref{thm:elements-equivalence} we obtain $\presheaves{\elements{\obs}} \cong \presheaves{\aseq\act}$. So, we take the category $\presheaves{\aseq\act}$ as the semantic universe to study fair transition systems. Just as in Example~\ref{ex:semantics-lts}, for a given fair transition system $\tuple{X,\act,\rightarrow,\fairness{X}}$, we define a presheaf
\[
\semantics{X}(\sigma) =  \left\{p\in\fexec{X} \mid \max \dom p =\sigma\right\} \quad (\text{for each $\sigma\in\aseq\act$}),
\]
and the action is given by restricting the domain of an execution. Moreover, for any fair simulation function $f$, we let $\semantics{f}_\sigma = f\circ \_$ (for each $\sigma\in\aseq\act$); thus, resulting in a presheaf semantics to fair transition systems.
%define a family of functions $\semantics f_{\sigma} p$ (for each $\sigma\in\aseq\act$): $\semantics{f}_\sigma p(\sigma') = fp\sigma'$, if $\sigma'\preceq\sigma \land \sigma'\in\fseq\act$; $\semantics{f}_\sigma p(\sigma) = \Omega$, if $\sigma\in\iseq\act$.
\begin{proposition}\label{prop:fair-semantics}
  The above map $\fts \rTo^{\semantics{\_}} \presheaves{\aseq\act}$ is a faithful functor.
\end{proposition}
Invoking the bisimulation map definition between any two presheaves generated by fair transition systems results in the characterisation of fair bisimulation. %Thus, presheaves are the right semantic structures for studying fair (bi)simulations.
\begin{theorem}\label{thm:fair-bisimulation-map}
  A fair simulation function $f$ is a fair bisimulation if, and only if, the underlying map $\semantics{f}$ is a bisimulation map in $\presheaves{\aseq\act}$.
\end{theorem}
%\begin{remark}

\noindent\stepcounter{thm}\textit{Remark~\thethm.}
  It should be noted that fair bisimulation maps are stronger than the open maps studied by Hildebrandt \cite{hildebrandt:fairness} in $\presheaves{\aseq\act}$ to characterise the extended bisimulation of Hennessy and Stirling \cite{hennessy:futureperfect} for pointed systems, i.e. systems with explicitly initial states. To demonstrate this, recall that an $\text{Inf}_\bot$-open map of Hildebrandt \cite[Proposition~25]{hildebrandt:fairness} is a fair simulation function (not necessarily surjective) $X\cup\{\Omega\} \rTo^f Y \cup \{\Omega\}$ satisfying \eqref{eq:bisim-reflect} and the property:
  \begin{equation}\label{eq:hildebrandt-fair}
    \forall_{x\in X,q\in\fexec Y}\ \big( q(\eseq)=f(x)  \implies \exists_{p\in\fexec{X}}\ (p(\eseq)=x \land f\circ p=q) \big).
  \end{equation}
  %\begin{wrapfigure}[2]{r}{4.5cm}%[7]{r}{5.4cm}
%  \vspace{-1.2cm}\hspace{-0.2cm}
%
%  \vspace{-0.4cm}
%  %\caption{Fair bisimulation maps are finer than open maps in \cite{hildebrandt:fairness}.}
%%  \label{fig:finer-fairbisim}
%	\end{wrapfigure}
Consider the two systems drawn on the right with a function $f$ between the states depicted by dashed lines.
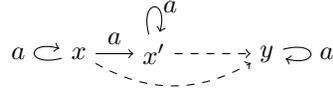
\[\scalebox{1}{
  \begin{tikzpicture}
    \node (x1) {$x$};
    \node (x2) at ($(x1.center)+(1,0)$) {$x'$};
    \node (x12) at ($(x2.center)+(1.5,0)$) {$y$};
    \path[->]
        (x1) edge node[above]{$a$} (x2)
        (x1) edge[loop left] node[left] {$a$} (x1)
        (x2) edge[loop above] node[right] {$a$} (x2)
        (x12) edge[loop right] node[right] {$a$} (x12);
    \path[->,dashed]
        (x1) edge[bend right] (x12)
        (x2) edge (x12);
  \end{tikzpicture}}
  \]
In the left system, the infinite executions visiting $x'$ infinitely often are considered fair, whereas the only infinite execution in the right system is fair. Clearly, $f$ is a fair simulation satisfying \eqref{eq:hildebrandt-fair}.
  But $f$ is \emph{not} a fair bisimulation because the sequence of finite executions $(p_i)_{i\in\nat}$ formed by unfolding the self-loop on $x$ has a fair execution $\bigsqcup_{i\in\nat} f\circ p_i$ (looping on $y$) as the limit in the right system. Yet, $(p_i)_{i\in\nat}$ has no limit, thus, violating \eqref{eq:fbisim-reflect}.

\begin{remark}
In \cite{hildebrandt:fairness}, (separated) presheaves with sup topology are used because an increasing sequence of finite executions induced by a fair transition system has at most one limit point. Unlike \cite{hildebrandt:fairness}, we are not interested in one-to-one semantic representation of our syntactical models, so we work with arbitrary presheaves. This is because the choice whether the semantic universe should be (separated) presheaves or sheaves depends on observations, i.e., how $\obs$ is modelled. This is also why we do \emph{not} require our semantic functor $\semantics{\_}$ to be full. %In particular, our hypothetical observer can observe all the infinite words, but it is the
\end{remark}

Next, we relate fair bisimulation maps with $\forall$-fair bisimulation relations.
\begin{definition}\label{def:forall-fairbisim}
A \emph{$\forall$-fair bisimulation} on $\tuple{X,\act,\rightarrow,\fairness {X}}$ is an equivalence relation $\mcal R\subseteq X\times X$ satisfying the following transfer properties:
\begin{enumerate}
  \item\label{def:forall-fairbisim1} $\forall_{x,y,x',a}\ \big( (x\step a x' \land x \mcal R y ) \implies \exists_{y'}\ (y \step a y' \land x' \mcal R y') \big)$, and
  \item\label{def:forall-fairbisim2} $\forall_{p,q}\
    \big( (p=_{\mcal R} q \land p\in\fairness X ) \implies q\in\fairness{X}\big)
    $.
  %\big( (x\mcal R y \land p=_{\mcal R} q \land p(\eseq)=x \land p\in\fairness X \land q(\eseq)=y) \implies q\in\fairness{X}\big) $.
\end{enumerate}
Here, $p=_{\mcal R} q$ is an abbreviation for $\dom q=\dom p \land \forall_{\sigma\in\dom p\cap\fseq\act}\ p(\sigma) \mcal R q(\sigma)$.
\end{definition}

%Now we establish the second result of this section.
\begin{theorem}\label{thm:existsfairisfairbisimfct}
Two states $x$ and $x'$ are related by a $\forall$-fair bisimulation relation if, and only if, there is a fair bisimulation function $f$ such that $f(x)=f(x')$. %For fairness predicates that are closed under the addition and the removal of finite prefixes, the converse holds as well.
\end{theorem}
The requirement of equivalence relations in the above definition may look superfluous at first glance. This is because, traditionally (i.e., when fairness predicates are empty sets), a strong bisimulation relation by definition is not necessarily an equivalence relation on the set of states. Moreover, \emph{bisimilarity} which is defined as the union of all strong bisimulation relations turns out to be both an equivalence relation and a strong bisimulation relation. However, such closure results do not hold in general for $\forall$-fair bisimulation relations. %In Appendix~\ref{subsec:fair-closure}, we show through counterexamples that
$\forall$-fair bisimulation relations are not closed under union and the relational composition (even if we relax Definition~\ref{def:forall-fairbisim} by replacing `an equivalence relation' for `a symmetric relation')\footnote{This is how $\forall$-fair bisimulation relations are defined in \cite{Kupferman2003:fair_equiv_rel} on the states of Kripke structures. In regards to the above closure properties, we are only aware of \cite{Hojati:96-phdthesis}
%To the best of our knowledge, such closure properties were not studied -- even for $\forall$-fair similarity -- in the model checking literature with the exception of \cite{Hojati:96-phdthesis}
who showed that $\forall$-fair simulation relations are not closed under union.}. Thus, $\forall$-fair bisimilarity may, in general, neither be an equivalence relation nor a $\forall$-fair bisimulation; in other words, $\forall$-fair bisimilarity is not a coinductive definition (like how strong bisimilarity is). Nevertheless, we deemed all $\forall$-fair bisimulation relations to be equivalences because: first, the main use of a $\forall$-fair bisimulation relation is to equate two systems that have the same behaviour sensitive to fair executions; second, the mathematics tells us that the kernel of a (fair bisimulation) function is an equivalence relation.

\section{The case of invisible actions}\label{sec:invsible}

The behaviour of a specification and its implementation is often spread over different levels of abstraction. The standard process algebraic way to relate the behaviour of an implementation with its specification is by delineating the effect of actions in lower levels of abstraction as invisible. For instance, removing a message from a buffer is considered unobservable in a rendezvous between communicating processes. This is made formal by reinterpreting the notion of (bi)simulation functions in the presence of the invisible action $\tau$.

\subsection{Branching bisimulation}
\begin{notation}
  Henceforth, $\tau\not\in\act$ will denote the invisible action and $\tact=\act\cup\{\tau\}$. Furthermore, $\exec{X,\sigma}$ is the set of executions $p$ having trace $\sigma$, i.e., $\max \dom p=\sigma$ and $\steps{} \subseteq X \times \fseq{\act} \times X$ is the \emph{weak reachability} relation on $\tuple{X,\tact,\rightarrow}$ %\todo{Added definition of weak reachability because of Reviewer (1) Remark (15)}
  given as the smallest relation satisfying the following conditions:
  \[
  \frac{ }{x\steps \eseq x} \qquad \frac{x\step\tau x'}{x\steps \eseq x'} \qquad \frac{x\steps\eseq x'\wedge x'\steps\eseq x''}{x\steps \eseq x''} \qquad \frac{x \steps\sigma x' \land x'\step a x''}{ x \steps{\sigma a}x''}.
  \]
\end{notation}
\begin{definition}\label{def:bsimulation}
  A \emph{branching simulation} function $f$ between systems $\tuple{X,\tact,\rightarrow}$ and $\tuple{Y,\tact,\rightarrow}$ is a function $X \rTo^f Y$ satisfying the following properties:
  \begin{enumerate}
    \item Simulation of observable transitions, i.e., $\forall_{x,x'\in X,a\in\act}\ x \step a x' \implies f(x) \step{a} f(x')$, and
    \item (Possible) simulation of invisible transitions, i.e., $\forall_{x,x'\in X}\ x \step \tau x' \implies  \big(f(x)=f(x') \lor f(x) \step{\tau} f(x')\big)$.
    \item Stuttering of $\tau$-transitions, i.e., for any $x_1,x_2,x_3$ we have
    \begin{equation}\label{eq:b-stutter}
      \big(x_1 \steps\eseq x_2 \steps\eseq x_3 \land f(x_1)=f(x_3)\big) \implies f(x_1)=f(x_2).
    \end{equation}
  \end{enumerate}
  A \emph{branching bisimulation} $f$ is a branching simulation surjection $f$ satisfying the `weak' reflection of transitions, i.e., for any $x\in X,y\in Y$, and $a\in\tact$ we have
  \begin{equation}
  \label{eq:trans-breflect}
  f(x) \step a y \implies \exists_{x',x''\in X}\ (x \steps {\eseq} x' \step a x'' \land f(x')=f(x) \land f(x'')=y).
  \end{equation}
\end{definition}
\begin{definition}
  Given a labelled transition system $(X,\tact,\rightarrow)$, then a symmetric relation $\mcal R\subseteq X \times X$ is a \emph{branching bisimulation} \cite{bbisim:1996} if, and only if, the following transfer property is satisfied
  \[
  \forall_{x_1,x_2,y_1,a\in\tact} \Big( (x_1 \step a x_2 \land x_1 \mcal R y_1) \implies (a=\tau \land x_2 \mcal R y) \lor \exists_{y,y_2}\ \big( y_1 \steps{\eseq} y \step a y_2 \land x_1 \mcal R y \land x_2 \mcal R y_2\big) \Big).
  \]
  Two states $x,x'\in X$ are branching bisimilar if there exists a branching bisimulation $\mcal R$ such that $x\mcal R x'$.
\end{definition}
We work with branching bisimulation functions (not relations) because of the following result (Theorem~\ref{thm:bbisim-kernel}), which is similar in spirit to Lemma~{2.7} proved by Caucal in \cite{caucal:1992}. The difference is that Caucal's branching bisimulation functions (which he calls \emph{reduction} in \cite{caucal:1992}) do not respect the stuttering of invisible steps \eqref{eq:b-stutter}. Nevertheless, we are still able to obtain the following correspondence since the largest branching bisimulation relation satisfies the so-called \emph{stuttering lemma} of \cite{bbisim:1996}. In particular, any reduction $X \rTo^f Y$ in the sense of Caucal can be extended to a branching bisimulation function by composing it with the quotient map $Y \rTo^q Y/\mcal R$, where $\mcal R\subseteq Y \times Y$ is the largest branching bisimulation relation on $Y$.
\begin{theorem}\label{thm:bbisim-kernel}
  Two states $x,x'\in X$ of a transition system $\tuple{X,\tact,\rightarrow}$ are branching bisimilar if, and only if, there are a transition system $\tuple{Y,\tact,\rightarrow}$ and a branching bisimulation function $X \rTo^f Y$ such that $fx=fx'$.
\end{theorem}
%\begin{remark}
%  We work with functions (not relations) because of the result proven in \cite{caucal:1992}: two states are branching bisimilar in the sense of van Glabbeek and Weijland \cite{bbisim:1996} if, and only if, they are mapped to the same state by a branching simulation function satisfying \eqref{eq:trans-breflect}.
%  Similarly, the close cousins of branching bisimulation delay, eta, and weak bisimulations can also be studied by modifying (1) and (2) in Definition~\ref{def:bsimulation}. However, we leave this aspect for future research.
%\end{remark}
%\begin{remark}
%  It might be a sur
%\end{remark}
\subsection{The setup}
Our first step towards the characterisation of branching bisimulation is the presheaf representation of executions induced by a transition system with invisible actions. Since a transition system evolves in a step-based manner, it is sufficient to model the time by the set of natural numbers $\nat$. The only difference, when compared to the strong case (cf.\thinspace Example~\ref{ex:observation-lts}), is in the notion of observation. Instead of recording words from $\fseq\act$, our hypothetical `observer' now records words from $\ftact$. So, consider a presheaf $\op\nat \rTo^{\A_\tau} \set$ in the spirit of $\A$ as defined in Example~\ref{ex:observation-lts}.
Recall that the slice category $\presheaves{\nat}/\A_\tau$ is equivalent to the category of presheaves on the category of elements $\elements{\A_\tau}$, which is isomorphic to the category $\ftact$. %Moreover, the categories $\elements{\A_\tau}$ and $\ftact$ are isomorphic since the length of a word is redundant in $\elements{\A_\tau}$ dropping which results in the elements of $\ftact$.
Thus, for a given transition  system $(X,\tact,\rightarrow)$, we define a presheaf $F_{X}\in\presheaves{\ftact}$ as follows:
\begin{equation}\label{eq:presheaf-tact}
  F_X(\sigma) = \exec{X,\sigma},\quad \text{(for every $\sigma\in\ftact$)}.
\end{equation}
The action on $F_{X}$ is given by the restriction of the domain of an execution.
%\begin{proposition}
%  The above mapping $F_X$ induced by $\tuple{X,\tact,\rightarrow}$ is a presheaf.
%\end{proposition}

Incidentally, a branching simulation function $X \rTo^f Y$ between $(X,\tact,\rightarrow)$ and $(Y,\tact,\rightarrow)$ does \emph{not} induce a system homomorphism between the underlying dynamical systems (presheaves) $F_X$ and $F_Y$ because a branching simulation function does not necessarily preserve the length of the executions. For example, a sequence of transitions $\bullet \step \tau \bullet \step a \bullet$ may get mapped to a transition $\bullet \step a \bullet$.

Thus, we need a procedure that transforms a given presheaf on $\ftact$ to a presheaf on $\fseq\act$. It turns out that there already is a general result in category theory for this purpose, which we explain next.
%\subsection{Hiding through adjunction}
In particular, recall the cocompletion of a category $\cat C$ through the Yoneda embedding $\cat C \rTo^{\yoneda{\cat C}} \presheaves{\cat C}$.
%One of the fundamental results in category theory is the Yoneda embedding $\cat C \rTo^{\yoneda{\cat C}} \presheaves{\cat C}$ of a category $\cat C$ into the category of presheaves $\presheaves{\cat C}$.
\begin{theorem}[see \cite{sheafbook}]\label{thm:cocompletion}
  For any functor $\cat C \rTo^h \cat D$, when $\cat C$ is small and $\cat D$ is cocomplete, there is a colimit preserving functor $\presheaves{\cat C} \rTo^{L_h} \cat D$ satisfying $L_h\circ \yoneda{\cat C} \cong h$. Moreover, $L_h$ has a right adjoint $R_h$ given by: $R_hD(C)=\cat D(hC,D)$, for each $C\in\cat C,D\in\cat D$.
  %\begin{diagram}
%    \cat C & \rTo^h & \cat D\\
%    \dTo^{\yoneda{\cat C}} & \ruTo_{L_h} &\\
%    \presheaves{\cat C}
%  \end{diagram}
\end{theorem}
Using the language of Kan extensions, $L_h$ is the left Kan extension of $h$ along the Yoneda embedding $\yoneda{\cat C}$. Moreover, using the notion of a coend \cite{categories-working-mat}, we have:
\[L_h F \cong \int^{C\in\cat C} FC \odot hC,\]
where $\cat D \rTo^{S \odot\_ } \cat D$ is the copower functor given by $S\odot D =\coprod_{s\in S} D$ (taking $S$ disjoint copies of $D$). If we replace $\cat C \rTo^h \cat D$ in the above diagram by a map $\cat C \rTo^h \cat D \rTo^{\yoneda{\cat D}} \presheaves{\cat D}$, then we obtain the so-called \emph{essential geometric morphism} \cite{sheafbook} between $\presheaves{\cat C} \rTo \presheaves{\cat D}$. In full, this means that the composition functor $h^*$ not only has a right adjoint $\Pi_h$, but also a left adjoint $\Sigma_h$ (see, for instance, \cite{sheafbook} for a formal definition of a geometric morphism). %Recall that a geometric morphism $\cat F \rTo^f \cat E$ between the two toposes is a pair of functors $\cat F \pile{\rTo^{f_*}\\ \lTo_{f^*}} \cat E$ such that $f^* \dashv f_*$ and $f^*$ is \emph{left exact} (it preserves finite limits).
Below, the composition functor $h^*$, its left adjoint $\Sigma_h$, and its right adjoint $\Pi_h$ are given by (up to isomorphism) $R_{\yoneda{\cat D}h}, L_{\yoneda{\cat D}h}$, and $R_{h^*\yoneda{\cat D}}$, respectively. % thus making $h^*$ as a left exact functor. %The following corollary states an application of the above fundamental result.
\begin{corollary}[see \cite{sheafbook}]\label{cor:ess_geometric}
  Given a functor $\cat C \rTo^h \cat D$ between small categories, then the inverse image functor $\presheaves {\cat D} \rTo^{h^*} \presheaves{\cat C}$ given by $h^*G=G\circ \op{h}$ (for each $G\in\presheaves{\cat D}$) has both left and right adjoints:
  \[\Sigma_h \dashv h^* \dashv \Pi_h.\]
\end{corollary}
So, in principle, there are two ways to land in $\presheaves{\fseq\act}$ from $\presheaves{\ftact}$ whenever there is a functor between $\ftact\rTo\fseq\act$. Nevertheless, since our aim is to characterise branching bisimulation functions, we will choose the left adjoint (over the right adjoint) of the composition functor for this task. This is because to reflect (recall \eqref{eq:trans-breflect}) an observable transition -- say, of the form $\bullet \step a \bullet$ --, we only need \emph{minimal} executions which are of the form $\bullet \steps{} \bullet \step a \bullet$ (cf.\thinspace Definition~\ref{def:mexecution}). In particular, we will show (cf.\thinspace Theorem~\ref{thm:adjoint-h}) that the left adjoint $\Sigma_h$ (induced by a suitable $h$) transforms a presheaf of executions into a presheaf of minimal executions (or those executions in which trailing $\tau$-transitions are chopped off).
%This situation induces another geometric morphism (albeit in the contra direction) if we further restrict to $\cat C,\cat D$ having finite limits and $h$ being left exact. As a side remark, this is another way to realise the adjunction between the inverse image and direct image sheaf functors modulo the sheafification process.
%\begin{corollary}[\protect{\cite[Example~4.1.10]{johnstone-vol1}}]\label{cor:geometric}
%  Recall Corollary~\ref{cor:ess_geometric} and further suppose that $\cat C,\cat D$ have finite limits and $h$ is left exact. Then $\Sigma_h$ is also left exact; thus, inducing a geometric morphism $\presheaves{\cat D} \rTo \presheaves{\cat C}$.
%\end{corollary}
\subsection{An obvious, but failed attempt}\label{subsec:failed}
There is an evident functor $\ftact \rTo^h \fseq\act$ which treats the letter $\tau$ as an empty word. It is then natural to investigate whether the bisimulation maps in $\presheaves{\fseq\act}$ between any two induced presheaves $\Sigma_h F_X,\Sigma_h F_Y$ characterise the branching bisimulation map. Unfortunately, the answer is no! We explain this extensively, as this lays the formal foundation for the desired characterisation given in the next subsection.

So, consider the hiding function $\tact \rTo^h \fseq\act$ which treats $\tau$ as the unit of $\fseq\act$ (i.e., $h(\tau)=\eseq$) and treats an action $a\in\act$ as observable (i.e., $h(a)=a$ when $a\in\act$).  This lifts to an order-preserving function $\ftact \rTo^h \fseq\act$ (denoted again as $h$ by abuse of notation). Thus, we have a functor $\ftact \rTo^h \fseq\act$.
\begin{proposition}\label{prop:A*-binary-products}
  The categories $\ftact,\fseq\act$ have binary products.%limits of nonempty small diagrams.
\end{proposition}
\begin{remark}
  It is worthwhile noting that $h$ does not generally preserve finite limits (or infimum $\sqcap$ in this case). Consider the words $a\tau b,ab\in\ftact$. Then, $h(a\tau b \sqcap ab) = h(a)=a$; however, $h(a\tau b) \sqcap h(ab) = ab \sqcap ab =ab$.
\end{remark}
\noindent
Furthermore, $S\odot X \cong S \times X$ (for any two sets $S,X$) and the (co)limit in any presheaf category is computed point-wise. Thus, using these observations, we calculate $\Sigma_h F$ (for $F\in\presheaves{\ftact}$) at $\varrho\in\fseq\act$ as follows:
\begin{equation}\label{eq:coend-simplification}
  \Sigma_h F (\varrho) \cong
  \int^{\sigma\in \ftact} F(\sigma)\times \yoneda {\fseq\act} (h (\sigma))(\varrho) \ \cong\
  \colimit{\sigma\in \op{\ftact},\varrho \preceq h(\sigma)} F(\sigma).
\end{equation}
Note that the reason for the last isomorphism is that every colimit can be encoded as a coend (see \cite[p.~224--225]{categories-working-mat} for a dual statement). %due to the encoding of the colimit of a diagram $\op{\cat I} \rTo^F \cat C$ by the coend of a bifunctor $\op{\cat I} \times \cat I \rTo^{F'} \cat C$ which is dummy in the second argument, i.e., $F'=F\circ \pi$, where $\op{\cat I} \times \cat{I} \rTo^\pi \op{\cat I}$ is the projection functor.
In addition, \eqref{eq:coend-simplification} is the colimit of $\op{\ftact} \rTo^{F} \set$ when $\varrho=\eseq$.
Notice that $\op{\ftact}$ is directed; thus, the colimit of the filtered diagram is $\coprod_{\eseq\preceq h(\sigma)} F\sigma/{\sim}$, where $\sim\ \subseteq \coprod_{\eseq\preceq h(\sigma)} F(\sigma)\times \coprod_{\eseq\preceq h(\sigma)} F(\sigma)$ is the equivalence relation defined as follows:
\begin{equation}\label{eq:filtered-colimit}
  (\sigma, p) \sim (\sigma',p') \iff \exists_{\sigma''}\ \left( \sigma''\preceq \sigma \land \sigma''\preceq \sigma' \land p\cdot \sigma'' =p'\cdot \sigma''\right).
\end{equation}
It turns out that from a system theoretic viewpoint, the set $\colimit{\eseq\preceq h(\sigma)}F_{X}(\sigma)$ is nothing but the set of minimal executions whose observable trace is an empty observation. Theorem~\ref{thm:adjoint-h} states this result in full generality.
\begin{definition}\label{def:mexecution}
  Given a transition system $\tuple{X,\tact,\rightarrow}$ and a word $\varrho\in\fseq\act$, we say an execution $p$ has an \emph{observable trace} $\varrho$ if $h(\max \dom p)=\varrho$. Furthermore, $\mexec{X,\varrho}$ is the set of all \emph{minimal executions} w.r.t.\thinspace the prefix order $\preceq$ on executions whose observable trace is $\varrho$. Formally, $p\preceq p' \iff p'|_{\dom p} = p$ and
  \[\mexec{X,\varrho} = \Big\{p \mid \ \history{p} \cap \{q \in\exec X \mid h(\max\dom{q})=\varrho\} = \{p\}\Big\}.\]
\end{definition}
%The next lemma, from a system theoretic viewpoint, says that the set of executions at each point in observation $\varrho\in\fseq\act$ in a presheaf $\Sigma_h\states X\in\presheaves{\fseq\act}$ is nothing but the set of minimal executions whose observable trace is $\varrho$.
Before we prove Theorem~\ref{thm:adjoint-h}, we need a category theoretic result (which is probably folklore; see \cite[Exercise~IV.2.7]{categories-working-mat} for a dual statement), namely, that the colimit of a diagram can be decomposed into the coproducts of the colimits of diagrams with smaller shapes under certain restrictions.
\begin{lemma}\label{lem:colimit-decompose}
  Given a set $J$ and a small category $\cat C=\coprod_{j\in J}\cat C_j$ with injections $\iota_j$, then for any functor $F$ from $\cat C$ to a cocomplete category $\cat D$, we have
  $\colimit{\cat C} F  \cong \coprod_{j\in J} \colimit{\cat C_j} F \circ \iota_j.$
\end{lemma}
\begin{theorem}\label{thm:adjoint-h}
  For a given transition system $\tuple{X,\tact,\rightarrow}$ and $\varrho\in\fseq\act$, we have
  $\Sigma_hF_X(\varrho) \cong \mexec{X,\varrho}.$
\end{theorem}
Next, we explore the action of the presheaf $\Sigma_hF_X$ from a system theoretic viewpoint. Firstly, it is defined by the universal property of the colimit. Let $\varrho'\preceq \varrho$ and let $\sigma$ be a minimal word such that $h(\sigma)=\varrho$. Then, there is a unique minimal word $\sigma_{\varrho'}=\bigsqcap \{\sigma'\preceq\sigma \mid h(\sigma')=\varrho'\}$. Note that this infimum exists since the history of a word in $\ftact$ is a finite totally ordered set of words. Therefore, the family of arrows depicted by the dotted arrows in \eqref{eq:sigma_h-action} forms a cone, where $\act_{\tau,\varrho}^\star=\{\sigma\in\ftact\mid \varrho \preceq h(\sigma)\}$ is a sub-forest of $\ftact$. Thus, the universal property of the colimit gives a map $\Sigma_hF_{X}\varrho \rTo \Sigma_hF_{Y}\varrho'$, which we denote by $\Sigma_h(\varrho,\varrho')$.

\vspace{-0.3cm}
\begin{minipage}[t]{0.4\textwidth}
  \begin{diagram}[width=1.6cm]
  F_X\sigma & \rAllMap & F_X\sigma_{\varrho'}\\
  \dTo & \eqnum\label{eq:sigma_h-action} & \dAllMap\\
  \colimit{\sigma\in\act_{\tau,\varrho}^\star} F_{X}\iota_{\varrho} \sigma & \rExistMap &\colimit{\sigma\in\act_{\tau,\varrho'}^\star} F_{X}\iota_{\varrho'} \sigma
 \end{diagram}
\end{minipage}
\hspace{0.8cm}
\begin{minipage}[t]{0.45\textwidth}
  \begin{diagram}[width=1.8cm]
    \Sigma_h F_{X} \varrho & \rTo~{\cong} & \mexec{X,\varrho}\\
    \dTo~{\Sigma_h(\varrho,\varrho')} & \eqnum\label{eq:naturality-mexec} & \dTo~{\mpast(\varrho,\varrho')}\\
    \Sigma_h F_{X} \varrho' & \rTo~{\cong} & \mexec{X,\varrho'}
  \end{diagram}
\end{minipage}
\vspace{0.1cm}

\noindent
Secondly, for any $\varrho' \preceq \varrho$ with $\varrho,\varrho'\in\fseq\act$, we define a map:
\[
\mexec{X,\varrho}\rTo^{\mpast(\varrho,\varrho')} \mexec{X,\varrho'} \quad \text{given by}\quad p \mapsto p|_{\history{\sigma_{\varrho'}}}\ ,
\]
%$\mpast(\varrho,\varrho')$ of type $\mexec{X,\varrho}\rTo \mexec{X,\varrho'}$:
where $\sigma=\max \dom p$. %$p_{\varrho'}=\bigsqcap \{p' \preceq p \mid p'\in \mexec{X,\varrho'}\}$. %Intuitively, $\mpast$ maps a minimal execution to a unique minimal execution in its history.
Now we can establish that the isomorphism in Theorem~\ref{thm:adjoint-h} is natural in $\varrho\in\fseq\act$.
\begin{theorem}\label{thm:square8commutes}
  For any $\varrho,\varrho'\in\fseq\act$ with $\varrho'\preceq\varrho$, the square in \eqref{eq:naturality-mexec} commutes.
\end{theorem}
Now that we know what $\Sigma_h$ does to a presheaf, we use the category $\presheaves{\fseq\act}$ as our semantic universe to handle invisible actions. Thus, for a given $\tuple{X,\tact,\rightarrow}$, we let $\semantics{X}=\Sigma_h F_X$. In addition, for a branching simulation function $X\rTo^f Y$, we use the isomorphic view of minimal executions and let $\semantics{f}_{\varrho}(p)=p_f$ (for each $\varrho\in\fseq\act$), where $p_f$ is an execution defined inductively using the following rules.
\begin{itemize}
  \item If $\dom p=\eseq$ then $\dom {p_f} = \eseq$ and $p_f(\eseq)=f(p(\eseq))$.
  \item If $p \step a p'$, $p_f \step a q$, $\last q=f(\last {p'})$, and $a\in\act$ then $p'_f =q$.
  \item If $p \step \tau p'$ and $f(\last p) = f(\last {p'})$ then $p'_f=p_f$.
  \item If $p \step \tau p'$, $f(\last p) \neq f(\last {p'})$, $p_f \step \tau q,\last q =f(\last {p'})$ then $p'_f = q$.
\end{itemize}
Here, $p \step a p'\iff \dom {p'}= \dom p a \ \land \ p\prec p'$; the function $\last{p}$ returns the last visited state by the execution $p$, i.e., $\last p = p(\max \dom p)$.
%\begin{remark}
%  Notice that the urge to define $\tilde f_\sigma$ as $\tilde f_\sigma(p) = f\circ p$ is simply hazardous, since for an execution $p\in \exec{X}$ it could be that the function $f\circ p$ is not even an execution in the system $Y$, i.e., $f\circ p\not\in\exec Y$.
%\end{remark}
\begin{lemma}\label{lem:mexec-preservation}
  For a given branching simulation function $X \rTo^f Y$ and a minimal execution $p\in\mexec{X,\varrho}$, we have $p_f\in\mexec{Y,\varrho}$.
\end{lemma}
Next let $\lts_\tau$ be the category of transition systems with possible invisible steps, which comprises of transition systems as objects and morphisms are the collection of identity functions and branching simulation functions. Thus, we can now give a presheaf semantics for the category $\lts_\tau$. %(the category of labelled transition systems with branching simulation functions as morphisms).
\begin{theorem}\label{thm:lts-tau_semantic1}
  The mapping $\lts_\tau \rTo^{\semantics{\_}} \presheaves{\fseq\act}$ is a faithful functor.
\end{theorem}

Finally, we can invoke the bisimulation maps in $\presheaves{\fseq\act}$ but the characterisation of branching bisimulation functions fails to hold. The mismatch is in the reflection \eqref{eq:trans-breflect} of invisible steps (not with the reflection of observable transitions), which we explain next in the following example.
%\begin{wrapfigure}[7]{r}{5cm}
%\vspace{-0.5cm}
%\end{wrapfigure}
%\vspace{-0.5cm}
\begin{example}
  Consider the two transition systems and the branching simulation $f$  depicted below with dashed lines. Note that $f$ is \emph{not} a branching bisimulation function (even though it is surjective) because it fails to reflect the transition $y_1 \step \tau y_3$.
  
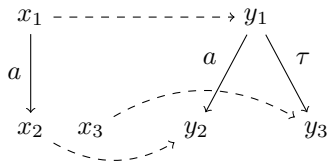
\begin{center}
  \begin{tikzpicture}
    \node (x1) {$x_1$};
    \node (x2) at ($(x1.center)+(0,-1.5)$) {$x_2$};
    \node (x3) at ($(x2.center)+(0.8,0)$) {$x_3$};
    \node (y1) at ($(x1.center)+(3,0)$) {$y_1$};
    \node (y3) at ($(y1.center)+(0.8,-1.5)$) {$y_3$};
    \node (y2) at ($(y1.center)+(-0.8,-1.5)$) {$y_2$};
    \path[->]
        (x1) edge node[left]{$a$} (x2)
        (y1) edge node[above right]{$\tau$} (y3)
        (y1) edge node[above left]{$a$} (y2);
    \path[->,dashed]
        (x1) edge (y1)
        (x2) edge [bend right] (y2)
        (x3) edge [bend left] (y3);
  \end{tikzpicture}
  \end{center}
  Yet $\semantics{f}$ is a bisimulation map in $\presheaves{\fseq\act}$. For this consider a given commutative square \eqref{eq:bisim-presheaf} in $\presheaves{\fseq\act}$. The crucial case in defining the functions %$Q(\varrho) \rTo^{k_\varrho} \mexec{X,\varrho}$
	$k_\varrho$ is when $\varrho\in\{\eseq,a\}$.
  \begin{itemize}
    \item Let $q\in Q(\eseq)$. If $n_\eseq (q)=\eseq_{y_i}$ then we let $k_\eseq(q)=\eseq_{x_i}$ (for $i\in\{1,2,3\}$).
    \item Let $q\in Q(a)$. Then, $n_a(q)$ is the minimal execution witnessing $y_1 \step a y_2$ since there are no other minimal executions in $\mexec{Y,a}$. Thus, we let $k_a(q)$ to be the minimal execution witnessing the transition $x_1 \step a x_2$.
  \end{itemize}
\end{example}

\subsection{Characterisation of branching bisimulation}
Unfortunately, all the information related to silent transitions is lost by transforming a presheaf using the functor $\Sigma_h$ because all the executions with empty observations get fused into their corresponding empty (minimal) executions. In retrospect, the problem lies in our specification of observation since our hypothetical `observer' $\obs$ is unable to differentiate between \emph{empty observation due to no system move} and \emph{empty observation due to zero or more system moves}. Therefore, to model this latter observation, we introduce a constant $\bar\tau$ which can be observed anytime after $0$. So for this section, we define $\obs$ in the next page as follows:

\begin{equation*}\begin{aligned}
   \obs(0) = \A_\tau(0) \quad \text{and} \quad \obs(n)=\A_\tau(n) \cup  \{\bar\tau\}\quad (\text{for every $n>0$});\\
   \sigma\cdot m=
\begin{cases}
  \sigma|_{m}, & \text{if}\ \sigma\neq\bar\tau\\
  \eseq, & \text{if}\ m=0 \land \sigma=\bar\tau  \\
  \bar\tau, & \text{if}\ m>0 \land \sigma=\bar\tau.
\end{cases}
\qquad \Big(\text{for any $m\leq n$ and $\sigma\in\obs(n)$}\Big).
\end{aligned}\end{equation*}

%Recall that the intuition behind the restriction of $\obs$ is to model how earlier observations are deduced from later observations, so we define:
%\[
%\sigma\cdot m=
%\begin{cases}
%  \eseq, & \text{if}\ m=0 \land \sigma=\bar\tau \\
%  \bar\tau, & \text{if}\ m>0 \land \sigma=\bar\tau \\
%  \sigma|_{m}, & \text{if}\ \sigma\neq\bar\tau.
%\end{cases}
%\qquad \text{(for any $m\leq n$ and $\sigma\in\obs(n)$)}.
%\]
\noindent
Put differently, $\bar\tau$ is `really' a constant observation over time except at $0$. The introduction of $\bar\tau$ is new and inspired from the constant $\eta$ \cite{eta-abstraction} which can be viewed as an \emph{empty observation due to at least one system move}.
%The introduction of $\bar\tau$ is our proposed analogon to the saturation based techniques for weak bisimilarity of which we will show that it reduces silent transitions to their observable aspects.

\begin{proposition}\label{prop:Oispresheaf}
  The above mapping $\op{\nat} \rTo^{\obs} \set$ is a presheaf.
\end{proposition}
\noindent
Moreover, the categories $\presheaves{\nat}/\obs$ and $\presheaves{\elements{\obs}}$ are equivalent (Theorem~\ref{thm:elements-equivalence}). So, consider the category of elements $\elements{\obs}$ whose objects are tuples $\tuple{n,\sigma}$ (for $n\in\nat,\sigma\in\obs(n)$) and whose arrows are given by the rule:
\[\tuple{n,\sigma} \step{} \tuple{n',\sigma'} \iff n\leq n' \land \sigma'\cdot n=\sigma.\]
Just as the category $\elements{\A_\tau}$ has a simpler description in terms of $\ftact$, we simplify $\elements{\obs}$ by defining a set $\ftbtact = \ftact \cup \{(n,\bar\tau) \mid n>0\} $ ordered by the smallest relation $\preceq\subseteq \ftbtact \times \ftbtact$ satisfying the following rules:
\[
\frac{(|\sigma|,\sigma) \step{} (|\sigma'|,\sigma') \in\elements{\A_\tau}}{\sigma \preceq \sigma'} \quad
\frac{(m,\bar\tau) \step{} (n,\bar\tau)\in\elements{\obs}}{(m,\bar\tau) \preceq (n,\bar\tau)}\quad
\frac{n>0}{\eseq \preceq (n,\bar\tau)}.
\]
%\[\forall_{\sigma,\sigma'\in\ftbtact}\ \sigma \preceq \sigma' \iff (|\sigma|,\sigma) \step{} (|\sigma'|,\sigma') \in \elements{\obs}.\]
%Note that in the above formulation we assumed that the length of $\bar\tau$ is $0$, i.e., $|\bar\tau|=0$.
Note that the leftmost rule is concerned with the elements of $\ftact$. Thus, the slice category $\presheaves{\nat}/\obs$ is equivalent to the category $\presheaves{\ftbtact}$ since $\elements{\obs}\cong\ftbtact$. Using the presheaf $F_X$ (cf. \eqref{eq:presheaf-tact}) induced by a given transition system $\tuple{X,\tact,\rightarrow}$, define a presheaf $\bar F_X\in\presheaves{\ftbtact}$ as follows:
\[\bar F_X\tuple{\sigma}=
\begin{cases}
  F_X(\sigma), & \text{if}\ \sigma\in\ftact\\
  \coprod_{n\in\nat} F_X(\tau^n), & \text{otherwise}
\end{cases},
\ \text{where}\
\begin{array}{l}
  \tau^0=\{\eseq\} \\
  \tau^{n+1}=\tau^n \cup \prod_{n+1}\{\tau\}
\end{array}.
\]
The executions based on invisible steps are seen \emph{stretchable in time} by $\obs$, i.e., a system may perform invisible executions of length independent of time instants. This is encoded in the second clause of $\bar F_X$. To complete the definition of $\bar F_X$ as a presheaf, we define, for any $p\in\bar F_X(\sigma)$ and $\sigma'\preceq \sigma$ in $\ftbtact$, the action of $\bar F_X$ as:
\[
p\cdot \sigma' =
\begin{cases}
  p|_{\history{\sigma'}}, & \text{if}\ \sigma\in\ftact \land \sigma'\in\ftact\\
  p|_{\{\eseq\}}, & \text{if}\ \sigma\not\in\ftact \land \sigma'=\eseq\\
  p, & \text{if}\ \sigma\not\in\ftact \land \sigma'\not\in\ftact.
\end{cases}
\]
Notice how the above formulation closely follows the three defining clauses of $\obs$.
\begin{proposition}\label{prop:barfxcfunctor}
  The mapping $\bar F_X$ defined above is a contravariant functor.
\end{proposition}
To characterise branching bisimulation, %we must remove the glass of invisibility from our observer, i.e.,
define a structure $\ftbact$ similar to $\fseq\act$ and a structure preserving map $\ftbtact \rTo \ftbact$ similar to $h$ (Section~\ref{subsec:failed}). To this end, our semantic category for branching bisimulation will be $\ftbact=\fseq\act\cup\{\bar\tau\}$ ordered by the relation $\preceq\subseteq\ftbact\times\ftbact$ consisting of prefix relation and the pairs $\tuple{\eseq,\bar\tau},\tuple{\bar\tau,\bar\tau}$. We consider $\ftbtact \rTo^{\bar h} \ftbact$ given by: $\bar h(\sigma)=h(\sigma)$, if $\sigma\in\ftact$; and $\bar h(n,\bar\tau)=\bar\tau$ (for all $n\in\nat$).
%\begin{proposition}
%  The categories $\ftbtact,\ftbact$ have limits and the above map $\bar h$ lifts to a finite limit preserving functor.
%\end{proposition}
%\begin{proof}
%  TODO.
%\end{proof}
%Thus, the conditions of Corollary~\ref{cor:geometric} are once again applicable and

Next, we aim to compute the presheaf $\Sigma_{\bar h}\bar F_X$ at $\varrho\in\ftbact$. At this stage, we can take the advantage of similarity between the structures $\ftact,\ftbtact$ and $\fseq\act,\ftbact$ to compute the colimits, similar to  \eqref{eq:coend-simplification} where $\op{\ftact}$ is replaced by $\op{\ftbtact}$ as the indexing category. In lieu of Theorems~\ref{thm:adjoint-h} and \ref{thm:square8commutes}, we obtain
\begin{theorem}\label{thm:adjoint-barh}
  For a given transition system $\tuple{X,\tact,\rightarrow}$, we find that $\Sigma_{\bar h}\bar F_X(\bar\tau) \cong \coprod_{n\in\nat} \exec{X,\tau^n}$ and $\Sigma_{\bar h}\bar F_X(\varrho) \cong \mexec{X,\varrho}$ (for any $\varrho\neq\bar\tau$).
%  \[
%  ,
%  \qquad
%   \quad (\text{}).
%  \]
  Moreover, the above isomorphisms are natural in $\varrho\in\ftbact$. I.e., for any $\varrho'\preceq\varrho$ with $\varrho,\varrho'\in\fseq\act$, the square in \eqref{eq:naturality-mexec} and the following square commute.
  \vspace{-0.3cm}
  \begin{diagram}[size=3em,width=2cm,midshaft]
			\Sigma_{\bar h} F_{X} (\bar\tau) & \rTo~{\cong} & \coprod_{n\in\nat} \exec{X,\tau^n}\\
			\dTo_{\Sigma_{\bar h}(\bar\tau,\eseq)} & & \dTo_{\_|_{\eseq}}\\
			\Sigma_h F_{X} (\eseq) & \rTo~{\cong} & \mexec{X,\eseq}
		\end{diagram}
\end{theorem}
  %\begin{wrapfigure}[5]{r}{5.4cm}\vspace*{-1.2cm}
%		
%	\end{wrapfigure}

Note that the definition of the maps $\Sigma_{\bar h}(\varrho',\varrho)$ (for $\varrho'\preceq \varrho$) is similar to the maps $\Sigma_{h}(\varrho',\varrho)$ defined by the universal property of colimits (see \eqref{eq:sigma_h-action}). Just like in the previous subsection, we utilise the isomorphic view of minimal executions to define our semantic map $\lts_\tau \rTo^{\semantics{\_}} \presheaves{\ftbact}$:
\begin{itemize}
  \item For a given transition system, we let $\semantics{X} = \Sigma_{\bar h}\bar F_X$.
  \item For a given branching simulation function $X \rTo^f Y$, we let $\semantics{f}_\varrho (p)=p_f$.
\end{itemize}
%Note that, just like in the previous subsection, we utilise the isomorphic view of minimal executions to define the map $\semantics{f}$.
\begin{lemma}\label{lem:semanticsfunctor}
  The above mapping $\semantics{\_}$ is a faithful functor.
\end{lemma}
Finally, we have obtained the desired result of this section.
\begin{theorem}\label{thm:bbisim-char}
  A branching simulation function $f$ is a branching bisimulation function iff $\semantics f$ is a bisimulation map in $\presheaves{\ftbact}$.
\end{theorem}

\section{Related work and Conclusion}
%We start out with a comparison of those related works that were not discussed in the introduction.

The core idea of \textbf{Goguen's sheaf semantics} \cite{Goguen92sheafsemantics} is: systems are diagrams of sheaves, behaviour (interconnection) of systems is their limit (colimit). In retrospect, Goguen gave a sheaf semantics of nondeterministic automata by constructing presheaves on the Alexandroff topology induced by the downward closed subsets of $\fseq\act$ (in contrast to presheaves over $\fseq\act$). However, it is well known that the category of sheaves on Alexandroff spaces induced by a poset $(X,\preceq)$ is equivalent to the category of presheaves on $X$ (a consequence of the so-called comparison lemma in topos theory; see \cite[Corollary~3 on Page~590]{sheafbook}). In short, we use simpler structures to represent executions and focused on defining bisimulations abstractly which were absent in \cite{Goguen92sheafsemantics}.

\textbf{Relational presheaves} \cite{rel-presheaves:2015} generalise transition systems that are labelled by words from the free monoid $\fseq\act$. The idea was to accommodate the earlier presheaf approaches \cite{presheaves-as-transitionsys:1997,jnw96:bisimopenmaps} with algebraic structure on labels. The most insightful observation of \cite{rel-presheaves:2015} was the well-known `saturation' construction on transition systems can be captured using a 2-adjunction induced by a homomorphism between %the underlying algebraic structures
$\ftact,\fseq\act$. %Actually, this construction was quite inspiring, which lead us to the explore the adjunction between inverse and direct image sheaf functors.
Unlike \cite{rel-presheaves:2015}, our left adjoint $\Sigma_h$ records the minimal executions induced by a transition system with silent steps. %; in particular, the stalks of $h^*F_X$ at $\varrho\in\fseq\act$ are nothing but the collection of minimal executions whose observable trace is $\varrho$.
We expect this to be relevant for probabilistic systems, where minimal executions are used to define a probability measure (cf. \cite{bk2017}).

\textbf{Open maps} between presheaves as defined in \cite{jnw96:bisimopenmaps} are instances of the open maps in a topos as introduced in \cite{openmaps-completness:1994}. This is because open maps between presheaves (as in \cite{jnw96:bisimopenmaps}) are natural transformations whose naturality square is a weak pullback in $\set$ (cf. \cite[Example~1.1]{openmaps-completness:1994}). We discarded the open maps between presheaves because they are incapable of establishing complete refinement between an implementation and its specficiation; though it is still interesting to assert whether the bisimulation maps (Def.~\ref{def:bisim-presheaf}) satisfy the axioms given in \cite{openmaps-completness:1994}.

\textbf{Prefix orders} are generalisations of trees proposed in \cite{Cuijpers:2013:DCM} to study executions of dynamical systems in an order theoretic manner. In \cite{openmaps-concrete:2015}, the authors defined functional bisimulation between prefix orders by reinterpreting the definition of open maps in concrete categories. Our bisimulation maps are an instance of this general definition (Section~\ref{sec:bisimmaps}). It is unclear, though, how to enrich prefix orders with observations so that we can model labelled executions in a uniformly. This  question lead us to model observations as presheaves.

To sum up, bisimulation maps between presheaves are versatile enough to capture
%We have seen that sheaves and open maps are versatile framework to model, in an abstract way,
different notions of behavioural equivalence.
We
%have
demonstrated
%the usefulness of this approach with
this by characterising $\forall$-fair bisimulation and branching bisimulation, two notions that are notoriously difficult to capture with a coalgebraic approach. The clear distinction between time and observation proved fruitful in dealing with silent actions, but we also expect our framework to lend itself well to modelling
%to model hybrid systems, e.g. timed automata.
hybrid systems. For instance, sheaves over the translation-invariant interval domain $\mathbb {IR}_{/\vartriangleright}$ were introduced in \cite{spivak_hybrid-systems} to model hybrid systems. It will be interesting to explore whether bisimulation maps between such sheaves (Remark~\ref{remark:sheaves-bisim}) coincides with stateless bisimulation \cite{hcif,Cuijpers-hypa,vanBeek:2006-chi}.
%Establishing a sheaf-theoretic model for timed automata is among the work we plan for the future.

%Another interesting extension -- taking inspiration from \cite{Barbara:2012-coalg-minimisation} -- %of our work we aim towards is to find an algorithm that allows constructing %open bisimulation maps in an iterative way, at least for discrete notions of time. Last, but not least, the issue of compositional semantics  in the spirit of \cite{sos-meta-theory:survey} still remains unsolved in this abstract setting.

%We have discussed some related work already in the respective sections. On a general note, (pre-)sheaf based modelling techniques have been investigated as an abstract modelling tool e.g. by Winskel et al. \cite{presheaves-as-transitionsys:1997,cattani_winskel_2005,hildebrandt:fairness,jnw96:bisimopenmaps,Fiore:wbisim_open-maps}, but the proposed notion of open maps used in this work can yield a clearer picture on the critical aspects of modelling the semantic aspects of systems. A related, yet different approach towards an abstract model for state-based systems is the coalgebraic line of research \cite{Rut03:universal}. When compared to coalgebra, the more flexible notion of time, in contrast to the strictly step-based semantics of coalgebras, is the main advantage of the open maps-approach proposed in this work.
%\printbibliography[]

%
% ---- Bibliography ----
%
% BibTeX users should specify bibliography style 'splncs04'.
% References will then be sorted and formatted in the correct style.
%
\begin{ack}
  We thank the anonymous reviewers of FOSSACS'19 for their feedback on an earlier draft that greatly improved this manuscript. In particular, Reviewer 1 not only identified a mistake in our earlier characterisation of branching bisimulation, but also provided a detailed feedback including the proofs in the appendix for which we are grateful.
  We thank Barbara K\"onig and Christina Mika for valuable comments on earlier drafts of this paper. We also thank Pieter Cuijpers for asking us about the importance of presheaf categories in system modelling, which formed the basis of Section~\ref{sec:prelim}. Moreover, we are grateful to Alex Simpson for discussing early ideas and in particular, for guiding our attention from sheaves on Alexandroff spaces induced by the poset $\fseq\act$ to just presheaves on $\fseq\act$. Finally, we acknowledge
Paul Taylor's diagram package for commutative diagrams.
\end{ack}
%\emph{Acknowledgements:}
\bibliographystyle{entcs}
\bibliography{references}
%
%%Appendix
%\end{document}
\newpage
\appendix
\setcounter{thm}{0}
\section{Proofs}
%\renewcommand{\thetheorem}{\thesection\arabic{theorem}}

%\begin{theorem}
%  Given two isomorphic categories $\cat C$ and $\cat D$, then there is an equivalence of categories between $\presheaves{\cat C},\presheaves{\cat D}$.
%\end{theorem}
%\begin{proof}
%  Let $\cat C \rTo^F \cat D$ and $\cat D \rTo^G \cat C$ be functors such that $G\circ F=\id{\cat C}$ and $F\circ G=\id{\cat D}$. Define two functors $\presheaves{\cat C} \rTo^{F'} \presheaves{\cat D}$ and $\presheaves{\cat D} \rTo^{G'} \presheaves{\cat C}$:
%  \[
%  F'P(D)= PG(D) \ \text{and}\ G'Q(C)= QF(C) \quad \text{for $C\in\cat C,D\in\cat D$}.
%  \]
%\end{proof}
\begin{theorem}[Theorem~\ref{thm:bisim-epi}]
  Every bisimulation map in a (pre)sheaf category is a retract; thus, every bisimulation map is an epimorphism.
\end{theorem}
\begin{proof}
In the commutative square \eqref{eq:bisim-presheaf}, take $P=0$ the initial object in $\presheaves{\cat C}$ (the proof for sheaves is similar), $Q=G$, $g=!_{G}$, $n=\id{G}$, and $m=!_{P}$. Clearly, $0 \rMono{!_{G}} G$ is the monomorphism since any two parallel arrows with initial object as the common codomain ought to be isomorphic. Thus, $f$ is a retract in $\presheaves{\cat C}$ since there is an arrow $k$ such that $f\circ k = \id{G}$. Hence, $f$ is an epimorphism.
\end{proof}

\section{Proofs Concerning Section~\ref{sec:fairness}}
\begin{proposition}[Proposition~\ref{prop:fair-semantics}]
  The above map $\fts \rTo^{\semantics{\_}} \presheaves{\aseq\act}$ is a faithful functor.
\end{proposition}
\begin{proof}
  It is clear that a given fair transition system $\tuple{X,\act,\rightarrow,\fairness{X}}$ induces a presheaf $\semantics{X}$ on $\aseq\act$. Let $X\cup \{\Omega\} \rTo^f Y\cup \{\Omega\}$ be a given fair simulation function and let $p\in\semantics{X}(\sigma)$. Then $f\circ p$ preserves the transition structure since $f$ satisfies \eqref{eq:trans-preserve}; thus,
  \[
  p(\sigma') \step a p (\sigma' a) \implies
  f(p(\sigma')) \step a f(p(\sigma'a))
  \]
  for any $\sigma' a\preceq \sigma$ if $\sigma\in\fseq\act$ (or $\sigma' a \prec \sigma$ if $\sigma\in\iseq\act$). Moreover, if $\sigma\in\iseq\act$, then $\semantics f_\sigma (p(\sigma))=\Omega$ by construction. Thus, $\semantics{f}_\sigma\circ p\in\semantics{Y}(\sigma)$.

  Furthermore, $f$ induces a natural transformation $\semantics{X} \rTo^{\semantics{f}}\semantics{Y}$ because for any $\sigma' \preceq \sigma$ and $p\in\semantics{X} (\sigma)$, we find that
  \[f \circ (p \cdot \sigma') = f \circ (p|_{\history{\sigma'}})= (f\circ p)|_{\history{\sigma'}}=(f\circ p) \cdot \sigma'.\]
  Lastly, to show that $\semantics{\_}$ is faithful, assume that $\semantics{f}=\semantics{g}$. Then we find that $\eseq_{f(x)}= \semantics{f}_\eseq (\eseq_x) = \semantics{g}_\eseq (\eseq_x)=\eseq_{g(x)}$, for any $x\in X$. Thus, $f=g$.
\end{proof}

\begin{theorem}[Theorem~\ref{thm:fair-bisimulation-map}]
  A fair simulation function $f$ is a fair bisimulation if, and only if, the underlying map $\semantics{f}$ is a bisimulation map in $\presheaves{\aseq\act}$.
\end{theorem}
\begin{proof}
  \fbox{$\Rightarrow$} Consider the commutative square as depicted below in $\presheaves{\aseq\act}$.
    \begin{equation}\label{eq:fbisim}
    \begin{diagram}
          Q & \rTo^n & \semantics{Y}\\
          \uMono^g & & \uTo_{\semantics{f}}\\
          P & \rTo_m & \semantics{X}
    \end{diagram}
    \end{equation}
    We will prove the existence of $Q \rTo^k \semantics{X}$ by well-founded induction. Note that the relation $\preceq$ is a well-founded relation on $\aseq\act$. To see this, let $U\subseteq\aseq\act$ and let $\sigma\in U$, for some $\sigma\in\aseq\act$. Since the history of $\sigma$ is discrete, i.e., $\history{\sigma}\setminus\{\sigma\}$ is isomorphic to $\history{n}$ (for some $n$ when $\sigma$ is finite) or $\nat$ (when $\sigma$ is infinite), we know that the infimum $\bigsqcap (\history{\sigma}\cap U)$ exists. Thus, every nonempty subset of $\aseq\act$ has a minimal element; whence, $\preceq$ is a well-founded relation on $\aseq\act$.

    For the base case, let $\sigma=\eseq$. Notice that the map $\semantics{f}_{\eseq}$ is surjective since the given $f$ is surjective. Moreover, since every epi splits in \set, there is a map $\semantics{Y}(\eseq) \rTo^h \semantics{X}(\eseq)$ such that $\semantics f_\eseq\circ h=\id{\semantics{Y}(\eseq)}$. Thus, we define
    \[
    k_\eseq(q)=
    \begin{cases}
      m_\eseq (\inv{g}_\eseq (q)), & \text{if}\  q\in g_\eseq (P(\eseq))\\
      h (n_\eseq (q)), & \text{otherwise}.
    \end{cases}
    \]
    Clearly, $\semantics{f}_\eseq\circ k_\eseq=n_\eseq$ and $k_\eseq\circ g_\eseq=m_\eseq$.

    For the inductive case, consider a family of functions $k_{\sigma'}$ (for each $\sigma'\preceq \sigma$) such that
\begin{enumerate}
  \item[$I.1$] $k$ is a natural transformation up to $\sigma$, i.e., $k_{\sigma''} (q\cdot \sigma'')= (k_{\sigma'} (q)) \cdot \sigma''$, for any $\sigma''\preceq\sigma' \preceq \sigma$ and $q\in Q(\sigma')$;
  \item[$I.2$] the lower triangle commutes up to $\sigma$, i.e., $k_{\sigma'}\circ g_{\sigma'} = m_{\sigma'}$ (for any $\sigma'\preceq\sigma$);
  \item[$I.3$] the upper triangle commutes up to $\sigma$, i.e.,
      $\semantics{f}_{\sigma'} \circ k_{\sigma'} = n_{\sigma'}$ (for any $\sigma'\preceq \sigma$).
\end{enumerate}
Next, we distinguish two cases based on the length of $\sigma$.
\begin{enumerate}
  \item Let $\sigma\in\fseq\act$. Then we know that $\sigma=\sigma' a$, for some $\sigma'\in\fseq\act,a\in\act$. Let $\bar q\in Q(\sigma)$ and assume that $\bar q\not\in g_{\sigma}( P(\sigma))$ (otherwise, simply take $k_{\sigma}(\bar q)=m_{\sigma}(\inv{g}_{\sigma} (\bar q))$). Then there are some $q'\in Q(\sigma')$ and $q\in \semantics{Y}(\sigma')$ such that $q'=\bar q\cdot \sigma'$, $n_{\sigma'}(q')=q$, and $q=(n_{\sigma}(\bar q))\cdot \sigma'$. And by the induction hypothesis we find an execution $p\in\exec{X,\sigma'}$  such that $k_{\sigma'}(q')=p$ and $\semantics{f}_{\sigma'}(p)=q$. Thus, $f(p) \prec n_\sigma(\bar q)$. And using \eqref{eq:bisim-reflect} we find a strict extension $p'$ of $p$ (i.e., $p'\prec p$) such that $f(p')=n_\sigma(\bar q)$. Clearly, $p'\cdot \sigma' = p$. Therefore, we let $k_{\sigma}(\bar q)=p'$. Clearly, $(k_{\sigma}(\bar q))\cdot \sigma' = k_{\sigma'}(\bar q \cdot \sigma')$. Moreover, $\semantics{f}_{\sigma}\circ k_{\sigma}(\bar q)=n_{\sigma}(\bar q)$; thus, the upper triangles commute up to $\sigma$. Lastly, commutativity of the lower triangle up to $\sigma$ follows directly from $k_{\sigma'}$. Hence, the inductive hypothesis ($I.1 \land I.2 \land I.3$) is satisfied by $k_\sigma$.
  \item Let $\sigma\in\iseq\act$. Let $\bar q\in Q(\sigma)$ and assume that $\bar q\not\in g_\sigma (P(\sigma))$. Then for every $i\in\nat$, there are some $q'_i \in Q(\sigma_i)$ and $q_i\in \semantics{Y}(\sigma_i)$ (where $\sigma_i=\sigma|_{i}\in\fseq\act$) such that $q_i'= \bar q\cdot\sigma_i$, $n_{\sigma_i}(q_i') = q_i$, and $q_i=(n_\sigma (\bar q))\cdot \sigma_i$. By induction hypothesis, for each $i$, we find some execution $p_i\in\exec{X}$ such that $k_{\sigma_i} (q_i')=p_i$ and $\semantics{f}_{\sigma_i} (p_i) = f(p_i) = q_i$. Note that this sequence $(p_i)_{i\in \nat}$ is strictly increasing because using the induction hypothesis $I.1$, we find (for any $i\in\nat$):
      \[
      p_{i+1}\cdot \sigma_i =
        (k_{\sigma_{i+1}}(q_{i+1}'))\cdot \sigma_i =
        k_{\sigma_i} (q_{i+1}'\cdot \sigma_i) = k_{\sigma_i} (q_i') = p_i.
      \]
      Thus, $p_i \prec p_{i+1}$. Clearly, the supremum $\bigsqcup_{i\in\nat} f(p_i)$ exists since it evaluates to the execution $n_{\sigma}(\bar q)$. So, from \eqref{eq:fbisim-reflect} we know that the supremum $\bigsqcup_{i\in\nat} p_i$ also exists and, moreover, $f(\bigsqcup_{i\in\nat} p_i) = n_\sigma(\bar q)$. Therefore, we let $k_\sigma(\bar q) = \bigsqcup_{i\in\nat} p_i$.

      Clearly, $(k_\sigma (\bar q))\cdot \sigma_i = p_i = k_{\sigma_i}(\bar q\cdot \sigma_i)$, for any $i\in\nat$. Moreover, $\semantics{f}_\sigma\circ k_\sigma (\bar q) = n_\sigma (\bar q)$; thus, the upper triangle commutes up to $\sigma$. Lastly, commutativity of the lower triangle up to $\sigma$ follows directly from each $k_{\sigma_i}$. Hence, the inductive hypothesis ($I.1 \land I.2 \land I.3$) is satisfied by $k_\sigma$.
\end{enumerate}
\fbox{$\Leftarrow$} Let $X\cup \{\Omega\}\rTo^f Y \cup \{\Omega\}$ be a fair simulation function such that $\semantics{f}$ is a bisimulation map in $\presheaves{\aseq\act}$. The function $f$ is surjective because the map $\semantics{f}_\eseq$ is surjective (cf. Theorem~\ref{thm:bisim-epi}). Next, we show that $f$ satisfies the condition in \eqref{eq:bisim-reflect}. Let $f(x) \step a y$ (for some $x\in X,a\in\act,y\in Y$) and let $q\in\exec{Y}$ witnessing this transition. Construct two presheaves $P,Q\in\presheaves{\aseq\act}$ as follows:
\begin{itemize}
  \item $P(\eseq)=\{\eseq_x\}$ and $P(\sigma)=\emptyset$ (for all $\sigma\neq\eseq$).
  \item $Q(\eseq)=\{\eseq_{f(x)}\}$, $Q(a)=\{q\}$, and $Q(\sigma)=\emptyset$ (for $\sigma\not\in \{\eseq,a\}$) with $q\cdot\eseq=\eseq_{f(x)}$.
\end{itemize}
Consider the square drawn in \eqref{eq:fbisim} and take $m,n$ as the obvious injections with $g_\eseq(\eseq_x)= \eseq_{f(x)}$ and $g_\sigma=\emptyset$ (for $\sigma\neq\eseq$). Then this square commutes, i.e., $\semantics{f}\circ m =n \circ g$. Moreover, the map $g$ is clearly a mono in $\presheaves{\aseq\act}$. Thus, there is a diagonal map $Q \rTo^k \semantics{X}$ such that the two triangles commute (see \eqref{eq:fbisim}). Therefore, there is an execution $p\in\semantics{X}(a)$ such that $p\cdot \eseq=\eseq_x$ and $f\circ p =q$ (due to the commutativity of lower and upper triangles, respectively). I.e., there is a transition $x \step a \last p$ and $f(\last p) = y$.

It remains to show that $f$ satisfies \eqref{eq:fbisim-reflect}. Consider an increasing sequence of finite executions $(p_i)_{i\in\nat}$. The interesting case is when the supremum $\bigsqcup_{i\in\nat} f\circ p_i$ exists. The other case when the supremum $\bigsqcup_{i\in\nat} p_i$ exists is trivial, since a simulation function already preserves such limit executions. Let $\sigma_i=\max \dom {p_i}$ (for each $i\in\nat$). Note that the supremum $\bigsqcup_{i\in\nat}\sigma_i$ exists since $\aseq\act$ is a directed complete poset. Now construct the following presheaves $P,Q\in\presheaves{\aseq\act}$.
\begin{itemize}
  \item $P(\sigma_i)=\{p_i\}$ and $P(\sigma)=\emptyset$ (for $\sigma\not\in \bigsqcup_{i\in\nat} \sigma_i$) with $p_i\cdot\sigma_j=p_j$ (for $j\leq i$).
  \item $Q(\sigma_i)=\{f\circ p_i\}$, $Q(\bigsqcup_{i\in\nat} \sigma_i)=\bigsqcup_{i\in\nat} f\circ p_i$, and $Q(\sigma)=\emptyset$ (for $\sigma\not\in \bigsqcup_{i\in\nat} \sigma_i$ or $\sigma\neq\bigsqcup_{i\in\nat}\sigma_i $) with $(\bigsqcup_{i\in\nat} f\circ p_i) \cdot\sigma_i = f\circ p_i$ and $(f\circ p_i)\cdot\sigma_j = f\circ (p_i\cdot\sigma_j)$. Notice that the last equation holds because $\semantics{f}$ is a natural transformation, which is due to $f$ being a fair simulation function (cf. Proposition~\ref{prop:fair-semantics}).
\end{itemize}
Furthermore, taking $m,n,$ and $g$ as the obvious injections in \eqref{eq:fbisim}, we again find that the square drawn in \eqref{eq:fbisim} commutes. Thus, there is a map $Q \rTo^k \semantics{X}$ such that the two triangles commute. Therefore, there is an execution $p\in\fexec X$ such that $p|_{\dom {p_i}} = p_i$ (for each $i\in\nat $) and $f\circ p =\bigsqcup_{i\in\nat} f\circ p_i$. Next, we claim that $p=\bigsqcup_{i\in\nat} p_i$. Clearly, $p$ is an upper bound of the sequence $(p_i)_{i\in\nat}$. Let $p'$ be another upper bound of this sequence. Then, $p'(\sigma_i) = p_i(\sigma_i) = p(\sigma_i)$ (for any $i\in\nat$). Moreover, $p'(\bigsqcup_{i\in\nat}\sigma_i) = \Omega = p(\bigsqcup_{i\in\nat}\sigma_i)$. Thus, $p=p'$; whence $p$ is the supremum of the sequence $(p_i)_{i\in\nat}$.
\end{proof}

\begin{theorem}[Theorem~\ref{thm:existsfairisfairbisimfct}]
Two states $x$ and $x'$ are related by a $\forall$-fair bisimulation relation if, and only if, there is a fair bisimulation function $f$ such that $f(x)=f(x')$.
\end{theorem}
\begin{proof}
	Let $\tuple{X,\act,\rightarrow,\fairness{X}}$ be a fair transition system.

\fbox{$\Rightarrow$} Assume a given fair bisimulation function $X\cup\{\Omega\} \rTo^f Y \cup\{\Omega\}$ with some fair transition system $\tuple{Y,\act,\rightarrow,\fairness{Y}}$. Then we show that the equivalence relation $\mcal R=\{(x,x')\mid f(x)=f(x')\}$ satisfies the transfer properties of a $\forall$-fair bisimulation relation.

Let $x\mcal Ry$ and $x\step ax'$. Using \eqref{eq:trans-preserve} we have $f(x)\step a f(x')$ and therefore $f(y)\step af(x')$, since $f(x)=f(y)$. Using \eqref{eq:bisim-reflect} we find some $y'\in X$ such that $y\step ay'$ and $f(y')=f(x')$. Thus, $x' \mcal R y'$.

Let $p\in\fairness{X}$ and let $q$ be an infinite execution such that $p=_{\mcal R} q$. Further, let $\sigma=\max\dom p$. Clearly, $\sigma\in\iseq\act$. Let $r_i=r\cdot (\sigma|_i)$, for each $i\in\nat$ and $r\in\{p,q\}$. Then we find $f\circ p_i = f\circ q_i$ because $p=_{\mcal R} q$. Now using \eqref{eq:fbisim-reflect}, we find $q=\bigsqcup_{i\in\nat} q_i\in\fairness{X}$ because
\[
f\circ \bigsqcup_{i\in\nat} p_i = \bigsqcup_{i\in\nat} f\circ p_i = \bigsqcup_{i\in\nat} f\circ q_i=f\circ \bigsqcup_{i\in\nat} q_i.
\]
\fbox{$\Leftarrow$} Assume a given $\forall$-fair bisimulation $\mcal R\subseteq X\times X$ which, by definition, is an equivalence relation. We construct a matching fair bisimulation function whose codomain is the quotient system $Y=\{[x]=\{x'\mid (x,x')\in \mcal R\}\mid x\in X\}$ with $[x]\step a[x']\iff \exists_{x_1\in[x], x_2\in[x']}\ x_1\step{a}x_2$. The fairness predicate on the quotient system $Y$ is defined as follows:
  \begin{equation}\label{eq:fairnessY1}
  \fairness{Y} = \{f\circ p \mid p\in\fairness{X}\},
  \end{equation}
  where $X \cup \{\Omega\} \rTo^f Y \cup \{\Omega\} $ by $f(x)=[x]$ and $f(\Omega)=\Omega$. We claim that the predicate $\fairness{Y}$ is well defined, i.e., independent of chosen representatives. To prove this, we show that
  \[\forall_{p,p'}\ \big(f\circ p =f\circ p' \implies (p\in\fairness{X} \iff p'\in\fairness{X}) \big).\]
  So let $p,p'$ be such that $f\circ p=f\circ p'$. Then we find that $\dom{p}=\dom{p'}$ and $\forall_{\sigma\in\dom p}\ p(\sigma) \mcal R p'(\sigma)$. I.e., $p=_{\mcal R} p'$. Clearly, from Definition~\ref{def:forall-fairbisim}\eqref{def:forall-fairbisim2} we obtain $p\in\fairness{X}\iff p'\in\fairness{X}$. Thus, $\fairness{Y}$ is well
  defined.

  It is straightforward to see that $f$ is a surjective fair simulation. The function $f$ satisfies \eqref{eq:trans-preserve} because $x \step a x' \implies [x] \step a [x']$. Moreover, $f$ preserves  fairness by the construction \eqref{eq:fairnessY1} of $\fairness{Y}$.

  Now to show that $f$ satisfies \eqref{eq:fbisim-reflect}, assume an increasing sequence of finite executions $(p_i)_{i\in\nat}$. The interesting case is when the supremum $q=\bigsqcup_{i\in \nat} f\circ p_i$ exists. The other case when the supremum $\bigsqcup_{i\in\nat} p_i$ exists is trivial since a fair simulation function preserves this limit. So consider the former case. Then define an infinite execution $\dom{q} \rTo^p X$ as follows:
      \[
      p(\sigma) =
      \begin{cases}
        p_{|\sigma|}(\sigma), & \mbox{if } \sigma\in\dom q\cap \fseq\act\\
        \Omega, & \mbox{otherwise}.
      \end{cases}
      \]
Note that $p$ is a well defined function since the domain of a fair infinite execution has one unique infinite word. Moreover, we know that $p$ is an infinite execution because of the given increasing sequence of finite executions.

Moreover, using \eqref{eq:fairnessY1}, we find some $p'\in\fairness{X}$ such that $q=f\circ p'$. Clearly, $\dom p =\dom q=\dom{p'}$. In addition,
\[
f (p(\sigma)) = f(p_{|\sigma|}(\sigma)) = q(\sigma) = f(p'(\sigma))\qquad \text{(for any $\sigma\in\fseq\act$)}.
\]
Clearly, $p(\max \dom q)=\Omega=p'(\max\dom{q})$ by the construction of an infinite execution. Thus, $p=_{\mcal R} p'$. So, from Definition~\ref{def:forall-fairbisim}\eqref{def:forall-fairbisim2}, we conclude that $p=\bigsqcup_{i\in\nat} p_i\in\fairness{X}$.

It remains to show that \ref{eq:bisim-reflect} holds. So, assume $f(x)\step ay$. By definition of $f$it follows that $[x]\step ay$, which means there exist $x_1\in [x], x_2\in y$ such that $x_1\step ax_2$. By Definition~\ref{def:forall-fairbisim}\eqref{def:forall-fairbisim1}, it follows then that there exists an $x'$ such that $x\step ax'\wedge x_2\in[x']$, i.e. $f(x_2)=f(x')=y$.%\todo{Added this proof part because of Reviewer (1), Remark (17)}\qed
\end{proof}

\subsection{Closure of $\forall$-fair bisimulation relations under relational composition and union}\label{subsec:fair-closure}
\begin{lemma}
  There are two $\forall$-fair bisimulation relations which are not closed under relational composition.
\end{lemma}
\begin{proof}
  Consider the following system $X=\{\Box_i,z\mid \Box\in\{x,y,z\} \land i\in\{1,2,3\}\}$:
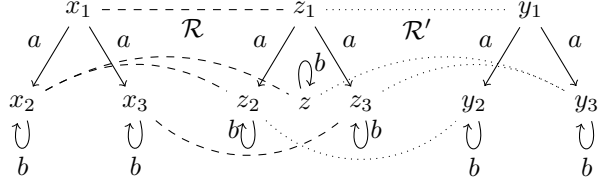
  \[
  \begin{tikzpicture}
    \node (x1) {$x_1$};
    \node (x2) at ($(x1.center)+(-0.75,-1.25)$) {$x_2$};
    \node (x3) at ($(x1.center)+(0.75,-1.25)$) {$x_3$};
    \node (z1) at ($(x1.center)+(3,0)$) {$z_1$};
    \node (z2) at ($(z1.center)+(-0.75,-1.25)$) {$z_2$};
    \node (z3) at ($(z1.center)+(0.75,-1.25)$) {$z_3$};
    \node (z) at ($(z1.center)+(0,-1.25)$) {$z$};
    \node (y1) at ($(z1.center)+(3,0)$) {$y_1$};
    \node (y2) at ($(y1.center)+(-0.75,-1.25)$) {$y_2$};
    \node (y3) at ($(y1.center)+(0.75,-1.25)$) {$y_3$};
    \path[->]
        (x1) edge node[above left]{$a$} (x2)
        (x1) edge node[above right]{$a$} (x3)
        (z1) edge node[above left]{$a$} (z2)
        (z1) edge node[above right]{$a$} (z3)
        (y1) edge node[above left]{$a$} (y2)
        (y1) edge node[above right]{$a$} (y3)
        (x2) edge[loop below] node[below] {$b$} (x2)
        (x3) edge[loop below] node[below] {$b$} (x3)
        (z2) edge[loop below] node[above left] {$b$} (z2)
        (z3) edge[loop below] node[above right] {$b$} (z3)
        (z) edge[loop above] node[right] {$b$} (z)
        (y2) edge[loop below] node[below] {$b$} (y2)
        (y3) edge[loop below] node[below] {$b$} (y3);
    \path[-,dashed]
        (x1) edge node[below] {$\mcal R$} (z1)
        (x2) edge[bend left] (z2)
        (x2) edge[bend left] (z)
        (x3) edge[bend right=40] (z3);
    \path[-,dotted]
        (z1) edge node[below] {$\mcal R'$} (y1)
        (z2) edge[bend right=44] (y2)
        (z) edge[bend left] (y3)
        (z3) edge[bend left] (y3);
  \end{tikzpicture}
  \]
  where dashed and dotted lines show the relations $\mcal R$ and $\mcal R'$, respectively. Furthermore, consider the following notion of fairness (note that $b^0=\eseq$):
  \[
  \fairness{X}=\{p \mid \text{$p$ is an infinite execution on $X$} \land \forall_{n\in\nat}\ p(ab^n)\in \{x_2,z_2,y_2\}\}.
  \]
  Clearly, $\mcal R$ and $\mcal R'$ are two $\forall$-fair bisimulation relations. However, their composition $\mcal R' \circ \mcal R$ is not a $\forall$-fair bisimulation. To see this, consider two infinite executions: $\history{a\iseq b} \rTo^{p,q} X$ with $p(\eseq)=x_1$, $q(\eseq)=y_1$ and $p(a b^n)=x_2$, $q(a b^n)=y_3$ (for all $n\in\nat$). Moreover, let $p(a \iseq{b})=\Omega=q(a \iseq{b})$. Clearly, $p=_{\mcal R'\circ \mcal R} q$ because $p(\eseq) \mcal R z \mcal R' q(\eseq)$ and $p(a b^n) \mcal R z \mcal R' q(a b^n)$ (for all $n\in\nat$). Yet, $p$ is a fair execution (i.e., $p\in\fairness{X}$) and $q$ is an unfair execution (i.e., $q\not\in\fairness{X}$). Thus the relation $\mcal R' \circ \mcal R$ violates Definition~\ref{def:forall-fairbisim}\eqref{def:forall-fairbisim2}. This shows that the lemma statement holds if we define $\forall$-fair bisimulation relations as in \cite{Kupferman2003:fair_equiv_rel}, i.e., without requiring the underlying relation to be an equivalence.

  Next, we show that how the above counterexample can be modified even if the underlying relations are equivalence relations. First, notice that both $\mcal R$ and $\mcal R'$ are symmetric relations (dashed lines and dotted lines are without any arrows). Consider the following extension of the two relations:
  \begin{equation*}
  \begin{aligned}
     \mcal T &=\ \mcal R \cup \id X \cup \{(z,z_2),(z_2,z)\} \\
     \mcal T'&=\ \mcal R' \cup \id X \cup \{(z,z_3),(z_3,z)\} .
  \end{aligned}
  \end{equation*}
  Clearly, $\mcal T$ is an equivalence relation; notice that the addition of pairs $(z,z_2)$ and $(z_2,z)$ in the relation $\mcal R$ results in a transitive and symmetric relation. Moreover, $z$ and $z_2$ are clearly $\forall$-fair bisimilar. On similar grounds, we can argue that $\mcal T'$ is a $\forall$-fair bisimulation relation. However, it is straightforward to see that the relation $\mcal T \mcal T'$ is not an equivalence relation because it is not transitive ($(x_1,z_1),(z_1,y_1)\in\mcal T \mcal T'$; however, $(x_1,y_1)\not\in\mcal T \mcal T'$). Thus, $\mcal T \mcal T'$ is not a $\forall$-fair bisimulation relation. \qed
\end{proof}
\begin{lemma}
There are two $\forall$-fair bisimulation relations which are not closed under union. This still holds if we restrict to Streett fairness constraints rather than arbitrary fairness constraints.
\end{lemma}
\begin{proof}
	Consider the following system $X=\{\Box_i\mid \Box\in\{x,y\}\land i\in\{1,2\}\}$:
	\[
    \begin{tikzpicture}
    \node (x1) {$x_1$};
    \node (x2) at ($(x1.center)+(1.5,0)$) {$y_1$};
    \node (x3) at ($(x2.center)+(3,0)$) {$x_2$};
    \node (x4) at ($(x3.center)+(1.5,0)$) {$y_2$};
    \path[->]
        (x1) edge node[above]{$a$} (x2)
        (x1) edge[loop left] node[left] {$a$} (x1)
        (x2) edge[loop right] node[right] {$a$} (x2)
        (x2) edge[bend left] node[below] {$a$} (x1)
				
        (x3) edge node[above]{$a$} (x4)
        (x3) edge[loop left] node[left] {$a$} (x3)
        (x4) edge[loop right] node[right] {$a$} (x4)
        (x4) edge[bend left] node[below] {$a$} (x3);
  \end{tikzpicture}
\]
whose fairness predicate is defined as follows:
$$\fairness{X}=\{p\mid |\{n\in\mathbb N\mid p(a^n)=x_i\}|=|\{n\in\mathbb N\mid p(a^n)=y_i\}|,i\in\{1,2\}\}.$$
Note that the symbol $p$ in $\fairness X$ is ranged over all the infinite executions induced by the  fair transition system above.
Informally, this means that the state $x_i$ is infinitely often reached if, and only if, $y_i$ (for $i\in\{1,2\}$) is infinitely often reached. It should be noted that this constraint can be seen to be a Streett constraint specified by the pairs $(x_1,y_1), (y_1,x_1), (x_2, y_2), (y_2, x_2)$.

% and if $x_2$ is infinitely often reached, then $y_2$ is infinitely often reached.

Next, we construct two $\forall$-fair bisimulation relations:
\begin{equation*}
\begin{aligned}
   \mcal R_1=&\ \{(x_1,y_2),(y_1,x_2),(y_2,x_1),(x_2,y_1)\} \\
   \mcal  R_2=&\ \{(x_1,x_2),(y_1,y_2),(x_2,x_1),(y_2,y_1)\}.
\end{aligned}
\end{equation*}
	First, we observe that $\mcal R_1$ and $\mcal R_2$ are both symmetric.
	We show for $\mcal R_2$ that also (1) and (2) of Definition~\ref{def:forall-fairbisim} hold; it can be shown completely analogously for $\mcal R_1$ (observe the symmetric nature of the system).
	\begin{enumerate}
		\item Assume $x_i\step ay_i$ then $x_{3-i}\step ay_{3-i}$ and $(y_i,y_{3-i})\in \mcal R_2$. If $x_i\step ax_i$ then $x_{3-i}\step ax_{3-i}$ and $(x_i,x_{3-i})\in \mcal R_2$. Analogously for transitions emanating from $y_i$.
		\item Assume $p\in\fairness{X}$, then for $i=1,2$, we find $|\{n\in\mathbb N\mid p(a^n)=x_i\}|=|\{n\in\mathbb N\mid p(a^n)=y_i\}|$. Now, if $q=_{\mcal R_2}p$, then
		$$q(\sigma)=\begin{cases}x_{3-i}&\text{if\ }p(\sigma)=x_i\\y_{3-i}&\text{if\ }p(\sigma)=y_i\end{cases},$$i.e. $|\{n\in\mathbb N\mid q(a^n)=x_i\}|=|\{n\in\mathbb N\mid p(a^n)=x_{3-i}\}|=|\{n\in\mathbb N\mid p(a^n)=y_{3-i}\}|=|\{n\in\mathbb N\mid q(a^n)=y_i\}|$, thus $q\in\fairness{X}$.
	\end{enumerate}
		Now we consider $\mcal R=\mcal R_1\cup \mcal R_2$ and the fair infinite execution $p$ such that $$p(\sigma)=\begin{cases}x_1&\text{if\ }|\sigma|\text{\ is\ even}\\y_1&\text{if\ }|\sigma|\text{\ is\ odd}\end{cases}.$$
		Moreover, we consider the unfair infinite execution $q$ such that $q(\sigma)=x_2$ for all finite words $\sigma$. Then for all finite words $\sigma$, we have $p(\sigma) \mcal R q(\sigma)$, because for all $\sigma$ of even length $p(\sigma)=x_1\ \mcal R_2\ x_2=q(\sigma)$ and for all $\sigma$ of odd length $p(\sigma)=y_1\ \mcal R_1\ x_2=q(\sigma)$. So $p=_{\mcal R}q$, but $q$ is not fair, so $\mcal R$ cannot be a $\forall$-fair bisimulation.
		
		Also note that $\mcal R_1$ and $\mcal R_2$ can be extended to equivalence relations by adding just the reflexive pairs (which trivially does not impact the bisimulation property), so this lemma even holds when restricted to equivalence relations.
\end{proof}

We will now demonstrate why the choice to restrict $\forall$-fair bisimulation to equivalences rather than merely to symmetric relations \footnote{It is easy to see that if we do not require $\forall$-fair bisimulations to be symmetric, we obtain a different notion: Just take two states $x$ and $y$ with a self loop each, where the (unique) infinite run from $y$ is fair and the one from $x$ is not, then $x$ could be related to $y$ but not vice-versa.} is sensible: $\forall$-fair bisimulations would not be closed under transitivity otherwise.%\todo{Added proof that transitivity is required; based on point (9) of reviewer (1)}
\begin{lemma}
	If we only require $\mcal R$ in the definition of $\forall$-fair bisimulation to be symmetric rather than an equivalence, then $\forall$-fair bisimulation is not closed under transitivity.
\end{lemma}
\begin{proof}
	Consider the following system $X=\{x_i^j\mid i,j\in\mathbb N\}$ over a single letter alphabet (the action is omitted):
	\begin{tikzpicture}
    \node (x11) {$x_1^1$};
    \node (x12) at(1,0){$x_1^2$};
    \node (x13) at(2,0){$x_1^3$};
    \node (x14) at(3,0){$\dots$};
    \node (x21) at(0,-1){$x_2^1$};
    \node (x22) at(1,-1){$x_2^2$};
    \node (x23) at(2,-1){$x_2^3$};
    \node (x24) at(3,-1){$\dots$};
    \node (x31) at(0,-2){$\vdots$};
    \node (x32) at(1,-2){$\vdots$};
    \node (x33) at(2,-2){$\vdots$};
    \node (x34) at(3,-2){$\ddots$};

    \path[->]
        (x11) edge (x21)
        (x21) edge (x31)
        (x12) edge (x22)
        (x22) edge (x32)
        (x13) edge (x23)
        (x23) edge (x33)
        (x11) edge (x22)
        (x22) edge (x33);
  \end{tikzpicture}
	
	Where all infinite executions, except for $x_1^1\rightarrow x_2^2\rightarrow x_3^3\rightarrow...$ are fair. Let $\mcal R=\{(x_i^j,x_i^{j+1}),(x_i^{j+1},x_i^j),(x_i^j,x_i^j)\mid i,j\in\mathbb N\}$, then $\mcal R$ clearly is a $\forall$-fair bisimulation  -- except for transitivity of $\mcal R$. The transitive closure $\mcal R'$ of $\mcal R$ is not a $\forall$-fair bisimulation, though, because $x_1^1\rightarrow x_2^1\rightarrow x_3^1\rightarrow...=_{\mcal R'}x_1^1\rightarrow x_2^2\rightarrow R_3^3\rightarrow...$, $x_1^1\rightarrow x_2^1\rightarrow x_3^1\rightarrow...$ is fair, but $x_1^1\rightarrow x_2^2\rightarrow R_3^3\rightarrow...$ is not.
\end{proof}

\section{Proofs Concerning Section~\ref{sec:invsible}}
\begin{proposition}[Proposition~\ref{prop:A*-binary-products}]
  The categories $\ftact,\fseq\act$ have binary products.
\end{proposition}
\begin{proof}
  %We start by showing that $\ftact$ and $\fseq\act$ have finite limits.
  Note that $\ftact=\fseq{(\act\cup\{\tau\})}$. Therefore, by showing that $\fseq\act$ for arbitrary $\act$ has binary products, we also show that $\ftact$ does. Moreover, in a poset viewed as a category, binary product corresponds to meet, so for any $\sigma,\sigma'\in\fseq\act$, we define $\sigma\sqcap\sigma'$ as the unique element whose history is $\history{\sigma}\cap\history{\sigma'}$. It remains to show that $\sqcap$ satisfies the universal property of product. Clearly, $\sigma\sqcap \sigma' \preceq \sigma$ and $\sigma \sqcap \sigma' \preceq \sigma'$ because $\history{\sigma}\supseteq \history{\sigma} \cap\history{\sigma'} \subseteq \history{\sigma'}$. Suppose $\sigma \succeq \bar \sigma \preceq \sigma'$. Then,
  $
  \history{\sigma}\supseteq \history{\bar\sigma}\subseteq \history{\sigma'} \implies
  \history{\bar\sigma} \subseteq \history{\sigma}\cap \history{\sigma'} \implies
  \bar\sigma \preceq \sigma\sqcap\sigma'.
  $
\end{proof}

\begin{lemma}[Lemma~\ref{lem:colimit-decompose}]
  Given a set $J$ and a small category $\cat C=\coprod_{j\in J}\cat C_j$ with the natural injections $\iota_j$, then for any functor $F$ from $\cat C$ to a cocomplete category $\cat D$, we have
  \[\colimit{\cat C} F  \cong \coprod_{j\in J} \colimit{\cat C_j} F \circ \iota_j.\]
\end{lemma}
\begin{proof}
  We are going to use the Yoneda embedding to derive $D\cong D'$ whenever $\cat D(D,\_)\cong \cat D(D',\_)$, for any two objects $D,D'\in\cat D$. Recall the abstract definition of a colimit from \cite{categories-working-mat}: the colimit $\colimit {\cat C} F$ is a representation of the functor $[\cat C,\cat D](F,\Delta\_)$, i.e., we have an isomorphism $\cat D(\colimit{\cat C}F,D)\cong [\cat C,\cat D](F,\Delta D)$ natural in $D\in\cat D$. Thus, we can derive
  {\allowdisplaybreaks
  \begin{align*}
    \cat D (\coprod_{j\in J}\colimit{\cat C_j} F\circ \iota_j, D) & \cong\
        \prod_{j\in J}\cat D (\colimit{\cat C_j} F\circ \iota_j, D) \\
    & \cong\ \prod_{j\in J} [\cat C_j,\cat D] (F\circ \iota_j,\Delta (D) \circ \iota_j)\\
    & \cong\ \prod_{j\in J} \int_{C\in \cat C_j} \cat D(F \iota_j C,\Delta(D)\iota_j C)\\
    & \cong \int_{j\in J,C\in\cat C_j} \cat D(F \iota_j C,\Delta(D)\iota_j C)\\
    & \cong\ \int_{C\in \cat C} \cat D(FC,\Delta D C)\\
    & \cong\ [\cat C,\cat D](F,\Delta D)\\
    & \cong\ \cat D(\colimit{\cat C} F, D).
  \end{align*}
  }
  Clearly, all the isomorphisms are natural in $D$; thus, we obtain the desired isomorphism from the Yoneda embedding.
\end{proof}

\begin{theorem}[Theorem~\ref{thm:adjoint-h}]
  For a given transition system $\tuple{X,\tact,\rightarrow}$ and $\varrho\in\fseq\act$, we have
  \[\Sigma_hF_X(\varrho) \cong \mexec{X,\varrho}.\]
\end{theorem}
\begin{proof}
  Notice that the result in \eqref{eq:coend-simplification} can be seen as the colimit of the diagram:
  \[\op{\act_{\tau,\varrho}^{\star}} \rEmbedding^{\iota_\varrho} \op{\ftact} \rTo^{F_{X}} \set,\]
  where $\act_{\tau,\varrho}^\star=\{\sigma\in\ftact\mid \varrho \preceq h(\sigma)\}$ is a sub-forest of $\ftact$. Clearly, we can decompose the ordered set in the following way. Here, $M(\varrho)$ is the set of minimal words $\sigma$ such that $h(\sigma)=\varrho$, i.e., $M(\varrho)=\{\sigma\in\ftact \mid \history{\sigma} \cap \inv{h}(\varrho)=\{\sigma\}\}$:
  \[\act_{\tau,\varrho}^{\star} \cong \coprod_{\sigma\in M(\varrho)} \{\sigma' \mid \sigma \preceq \sigma'\}.\]
  This decomposition is due to the downward totality of the set $\ftact$. Lemma~\ref{lem:colimit-decompose} is now applicable and, upon applying, results in the following isomorphism:
  \[
  \colimit{\sigma\in\act_{\tau,\varrho}^\star} F_X\iota_\varrho \sigma
  \cong
  \coprod_{\sigma\in M(\varrho)} \colimit{\sigma'\succeq\sigma}F_X(\sigma').
  \]
  Notice that the index category $\{\sigma' \mid \sigma'\succeq \sigma\}$ (for each $\sigma\in\ftact$) is a directed set; thus, the colimit $\colimit{\sigma'\succeq\sigma}F_{X}(\sigma')$ can be computed using \eqref{eq:filtered-colimit}. Next, we claim that for any $\sigma\in M(\varrho)$, we have the following isomorphism:
  \begin{equation}\label{eq:exec}
  \colimit{\sigma'\succeq\sigma}F_{X}(\sigma') \cong \exec{X,\sigma}.
  \end{equation}
  Next, define two maps $\colimit{\sigma'\succeq\sigma}F_{X}(\sigma') \pile{\rTo~f \\ \lTo~g} \exec{X,\sigma}$:
  \[
  f[p,\sigma']_\sim = p\cdot \sigma \ \text{(for $p\in\exec{X,\sigma'}$)} \qquad
  g(p) = [p,\sigma]_\sim \ \text{(for $p\in\exec{X,\sigma}$)}.
  \]
  The function $f$ is well-defined because $(p_1,\sigma_1) \sim (p_2,\sigma_2) \implies \exists_{\sigma_3} (\sigma_1\succeq\sigma_3 \preceq \sigma_2 \land p_1\cdot \sigma_3 = p_2\cdot \sigma_3) \implies p_1\cdot\sigma = p_2\cdot\sigma$ (the last implication is because $\sigma\preceq\sigma_3$). Moreover, $f(g(p))=f([p,\sigma]_\sim) = p\cdot\sigma=p$ and $g(f([p,\sigma']_\sim)) =g(p\cdot\sigma) = [p,\sigma']_\sim$ because $(p\cdot\sigma)\cdot\sigma=p\cdot\sigma \implies (p\cdot\sigma,\sigma) \sim (p,\sigma')$.

  Thus, the isomorphism in \eqref{eq:exec} holds.
  Lastly, using \eqref{eq:mexec-char}, it is clear that $\coprod_{\sigma\in M(\varrho)} \exec{X,\sigma} \cong \mexec{X,\varrho}$.
\end{proof}

\begin{theorem}[Theorem~\ref{thm:square8commutes}]
  For any $\varrho,\varrho'\in\fseq\act$ with $\varrho'\preceq\varrho$, the square in \eqref{eq:naturality-mexec} commutes.
\end{theorem}
\begin{proof}
Recall the construction of the map $\Sigma_h(\varrho,\varrho')$ (for $\varrho'\preceq\varrho$) and recall the following diagram \eqref{eq:sigma_h-action}. Note that $\sigma_{\varrho'}=\bigsqcap \{\sigma'\preceq\sigma \mid h(\sigma')=\varrho'\}$ is a unique minimal word $\sigma'$ in the history of $\sigma$ with $h(\sigma')=\varrho$.

\noindent
Furthermore, from the proof of Theorem~\ref{thm:adjoint-h}, we know that
$$\colimit{\sigma\in A_{\tau,\varrho}^*} F_X(\iota_\varrho(\sigma)) \cong
\coprod_{\sigma\in M(\varrho)}\colimit{\sigma'\succeq\sigma}F_X(\sigma')
\cong
\coprod_{\sigma\in M(\varrho)} F_X(\sigma).$$
Thus, the injections of the above colimit (or the vertical arrows in the square drawn below) map an execution $p\in F_X(\sigma)$ (for some $h(\sigma)=\varrho$) to a minimal execution $p\cdot\sigma_\varrho\in F_X(\sigma_{\varrho})$.
\begin{diagram}[width=2.3cm]
  F_X(\sigma) & \rAllMap & F_X(\sigma_{\varrho'})\\
  \dTo & & \dAllMap\\
  \colimit{\sigma\in\act_{\tau,\varrho}^\star} F_{X}(\iota_{\varrho} (\sigma)) & \rExistMap &\colimit{\sigma\in\act_{\tau,\varrho'}^\star} F_{X}(\iota_{\varrho'} (\sigma))
 \end{diagram}
Consider an execution $p\in F_X(\sigma)$. Clearly, $\sigma_{\varrho'} \preceq \sigma$ and using \eqref{eq:mexec-char}, we find that $p\cdot \sigma_{\varrho'}\in \coprod_{\sigma\in M(\varrho)} F_X(\sigma)$. From the commutativity of the above square, we obtain
\begin{equation}\label{eq:sigma_h-mexec}
  p\cdot\sigma_{\varrho'} = \Sigma_h(\varrho,\varrho')(p\cdot\sigma_\varrho).
\end{equation}
Thus, we can simplify the action $\Sigma_h(\varrho,\varrho')$ in terms of a map:
\[\coprod_{\sigma\in M(\varrho)} F_X(\sigma) \rTo^{f(\varrho,\varrho')} \coprod_{\sigma\in M(\varrho')} F_X(\sigma),\ \text{where $f(\varrho,\varrho')(p,\sigma) = (p\cdot {\sigma_{\varrho'}},\sigma_{\varrho'})$}.\]
Next, we further simplify the elements of the set $\coprod_{\sigma\in M(\varrho)} F_X(\sigma)$:
\begin{equation}\label{eq:mexec-char}
  \forall_p\ \Big(p\in\mexec{X,\varrho} \iff \max\dom p \in M(\varrho)\Big).
\end{equation}
\begin{itemize}
  \item[$\Leftarrow$] Let $\max \dom p\in M(\varrho)$, for some $p$. Further, suppose
      \[p'\in\downarrow p\cap\{q\in\exec X\mid h(\max\dom{q})=\varrho\}.\] Clearly, $p'\preceq p$ and $h(\max\dom{p'})=\varrho$. Then $\max\dom{p'} \preceq \max\dom{p}$ and since $\max\dom p\in M(\varrho)$, we find that $\max\dom {p'} = \max\dom p $. Thus, $p'=p$.
  \item[$\Rightarrow$] Let $p\in\mexec{X,\varrho}$ and let $\max \dom p=\sigma$, for some $\sigma$. Then $\history{p} \cap \{q\in\exec{X}\mid h(\max \dom q)=\varrho\}= \{p\}$. Moreover, $\history{\sigma}\cap \inv{h}(\varrho) \neq\emptyset$ because $h(\sigma)=\varrho$. So, let $\sigma'\preceq\sigma$ such that $h(\sigma')=\varrho$. Then $p\cdot\sigma' \preceq p$. However, since $p$ is minimal, we find that $p\cdot\sigma'=p$, i.e., $\sigma'=\sigma$.
\end{itemize}
%By the universal property of colimit, we find for any given word $\sigma$ and $\varrho\preceq h(\sigma)$, using the function $\iota_\varrho$ the follwing arrow:
%$$F\sigma\rTo \colimit{\sigma\in A_{\tau,\varrho}^*}F_X\iota_\varrho(\sigma)\cong\coprod_{\sigma\in M(\varrho)}\colimit{\sigma'\geq\sigma}F_X\sigma'$$ which yields the commutative diagram:
%
%  \begin{tikzpicture}
%    \node (tl) at (0,0) {$F\sigma$};
%    \node (bl) at (0,-3) {$F\sigma_{\varrho'}$};
%    \node (tr) at (3,0) {$\coprod_{\sigma\in M(\varrho)}\colimit{\sigma'\geq\sigma}F_X\sigma'$};
%    \node (br) at (3,-3) {$\coprod_{\sigma'\in M(\varrho')}\colimit{\sigma''\geq\sigma'}F_X\sigma''$};
%    \path[->]
%        (tl) edge node[left]{$\cdot\sigma_{\varrho}$} (bl)
%        (tr) edge node[left]{$\cdot\sigma_{\varrho}$} (br)
%        (tl) edge node[left]{} (tr)
%        (bl) edge node[left]{} (br);
%  \end{tikzpicture}
\begin{diagram}[width=1.6cm]
  \coprod_{\sigma\in M(\varrho)}F_X(\sigma) & \rTo^{\pi_\varrho} & \mexec{X,\varrho}\\
  \dTo^{f(\varrho,\varrho')} & & \dTo_{\mpast(\varrho,\varrho')}\\
  \coprod_{\sigma\in M(\varrho')}F_X(\sigma) & \rTo^{\pi_{\varrho'}} & \mexec{X,\varrho'}
 \end{diagram}
To prove this theorem it remains to show that the above square commutes where $\pi_{\_}$ forgets the domain of a minimal execution. This is trivial since $\mpast(\varrho,\varrho')\circ \pi_\varrho(p,\sigma) = p|_{\history{\sigma_{\varrho'}}} = p\cdot\sigma_{\varrho'} = \pi_{\varrho'} \circ f(\varrho,\varrho') (p,\sigma)$.
\end{proof}

\begin{lemma}[Lemma~\ref{lem:mexec-preservation}]
  For a given branching simulation $X \rTo^f Y$ and a minimal execution $p\in\mexec{X,\varrho}$, we have $p_f\in\mexec{Y,\varrho}$.
\end{lemma}
\begin{proof}
  Let $p\in\mexec{X,\varrho}$. We prove this lemma by performing structural induction on $p$.
	\begin{itemize}
		\item If $\dom p=\eseq$, then $\dom{p_f}=\eseq$, $p_f(\eseq)=f(p(\eseq))$ is trivially minimal.
		\item If $p' \step a p$ and $a\in A$, then $p_f' \step a q$ and $\last q = f(\last{p})$, thus, $p_f=q$. By induction hypothesis, $p_f'$ is minimal at $\varrho'=h(\max\dom{p'})$ (i.e., $p_f'\in\mexec{Y,\varrho'}$) and since $p$ ends with a transition labelled $a\in A$, we know that $p_f$ is also minimal at $\varrho' a$.
		\item If $p'\step\tau p$ then $p$ and $p'$ have the same observable trace. Yet, this case is inapplicable since $p'\prec p$ and thus, $p$ is not minimal.
	\end{itemize}
\end{proof}

\begin{theorem}[Theorem~\ref{thm:lts-tau_semantic1}]
  The above mapping $\lts_\tau \rTo^{\semantics{\_}} \presheaves{\fseq\act}$ is a faithful functor.
\end{theorem}
\begin{proof}
  Lemma~\ref{lem:mexec-preservation} ensures that for any $p\in\Sigma_h F_X\varrho$, we also have $p_f\in \Sigma_h F_Y \varrho$.
	We use the isomorphic view and show that $\semantics f$ is a natural transformation, i.e., for any $\varrho'\preceq\varrho$, we have $\mpast(\varrho,\varrho')\circ\semantics f_\varrho=\semantics f_{\varrho'}\circ\mpast(\varrho,\varrho')$. So let $p\in\mexec{X,\varrho}$ with $\sigma=\max \dom p$ and $\sigma'=\max \dom{p_f}$. Then
	\begin{align}
		&\ \mpast(\varrho,\varrho')\circ\semantics f_\varrho(p)=\semantics f_{\varrho'}\circ\mpast(\varrho,\varrho')(p) \notag \\
		\iff &\ \mpast(\varrho,\varrho')(p_f)=\semantics f_{\varrho'}(p\mid_{\history\varrho'}) \notag \\
		\iff &\ %p_f\mid_{\history\varrho'}=(p\mid_{\history\varrho'})_f
p_f\cdot \sigma'_{\varrho'}=(p \cdot\sigma_{\varrho'})_f. \label{eq:mpast-simplification}
	\end{align}
	This we can show by structural induction on $p$. For the base case, let $p=\eseq_x$. Then $\varrho'=\varrho=\eseq$ and $p_f\cdot \sigma'_{\varrho'}=\eseq_{f(x)}= (p \cdot\sigma_{\varrho'})_f$. For the inductive case, let $p' \step a p$ with $p\in\mexec{X,\varrho}$ and $\max\dom{p'}=\sigma$. Clearly, $a\in A$, otherwise, if $a=\tau$, then $p$ is not minimal. Note that $(\sigma a)_{\varrho'}=\sigma_{\varrho'}$ and thus,
\begin{equation}\label{eq:ltstau-faitfhul-1}
  p\cdot(\sigma a)_{\varrho'} = p|_{\history{(\sigma a)_{\varrho'}}} = p|_{\history{\sigma_{\varrho'}}} = p'|_{\history{\sigma_{\varrho'}}} = p' \cdot\sigma_{\varrho'}.
\end{equation}
Note that $p'$ may not be a minimal execution at $\varrho'$; yet, $p'\cdot \sigma_{\varrho'}\in \mexec{X,\varrho'}$ is minimal due to Equation~\eqref{eq:mexec-char}. Let $\sigma'=\max\dom{(p'\cdot\sigma_{\varrho'})_f}$. From Lemma~\ref{lem:mexec-preservation}, we know that $(p'\cdot\sigma_{\varrho'})_f\in\mexec{Y,\varrho'}$. Then, by the induction hypothesis and substituting $p$ by $p'\cdot\sigma_{\varrho'}$ in \eqref{eq:mpast-simplification} we obtain:
\begin{equation}\label{eq:ltstau-faithful-IH}
  (p'\cdot\sigma_{\varrho'})_f\cdot \sigma'_{\varrho'}=((p'\cdot\sigma_{\varrho'}) \cdot\sigma_{\varrho'})_f = (p' \cdot\sigma_{\varrho'})_f
  %p'_f\cdot \sigma'_{\varrho'}=(p' \cdot\sigma_{\varrho'})_f.
\end{equation}
Moreover, by the definition of $p_f$, we find that $(p'\cdot\sigma_{\varrho'})_f \steps{\tau^n} p'_f \step a p_f$ with $\dom{p_f}=\sigma'\tau^n a$, for some $n\in\nat$. Note that $p_f'=(p'\cdot\sigma_{\varrho'})_f$, whenever $n=0$. Then we find
$(\sigma'\tau^n a)_{\varrho'}=\sigma'_{\varrho'}$ because $\sigma'\prec \sigma' \tau^n a$. Thus, we have
\[
p_f\cdot(\sigma'\tau^n a)_{\varrho'} =
 p_f|_{\history{(\sigma'\tau^n a)_{\varrho'}}} =
 p_f |_{\history{\sigma'_{\varrho'}}} =
 (p_f'\cdot \sigma_{\varrho'})_f |_{\history{\sigma'_{\varrho'}}} =
 (p_f'\cdot \sigma_{\varrho'})_f \cdot \sigma'_{\varrho'}.
%p'_f\cdot \sigma'_{\varrho'}=p'_f |_{\history{\sigma'_{\varrho'}}} = p_f|_{\history{(\sigma'a)_{\varrho'}}}= p_f\cdot (\sigma'a)_{\varrho'}.
\]
Using the above equation, \eqref{eq:ltstau-faitfhul-1}, and \eqref{eq:ltstau-faithful-IH}, we can conclude that
\[(p\cdot(\sigma a)_{\varrho'})_f =(p' \cdot\sigma_{\varrho'})_f = (p_f'\cdot \sigma_{\varrho'})_f \cdot \sigma'_{\varrho'} =p_f\cdot(\sigma'\tau^n a)_{\varrho'}.\]
%
%If $\varrho'=\epsilon$, then $\dom{p\mid_{\history\varrho'}}=\{\epsilon\}$ and thus $f(p\mid_{\history\epsilon}(\epsilon))=f(p(\epsilon))$ and $\dom{p_f\mid_{\history\epsilon}}=\{\epsilon\}$, $p_f\mid_{\history\epsilon}(\epsilon)=f(p(\epsilon))$.
%	
%	Now assume that $p_f\mid_{\history\varrho'}=(p\mid_{\history\varrho'})_f$ via induction and show that $p_f\mid_{\history\varrho'a}=(p\mid_{\history\varrho'a})_f$ for $\varrho'a\preceq\varrho$, $a\in\tact$. We only need to show that the last state is the same and for this we observe via the definition of $p_f$: $$p_f\mid_{\history\varrho}\step ap_f\mid_{\history \varrho'a}\quad p\mid_{\varrho'}\step a p\mid_{\history\varrho'a}$$
%	and thus $p_f(\varrho'a)=f(p(\varrho'a))\quad \last{(p\mid_{\varrho'a})_f}=f(\last{p\mid_{\history\varrho'a}})=f(p(\varrho'a))$.

  The faithfulness of $\semantics{\_}$ is similar to the strong case since empty executions are always minimal executions.
\end{proof}

\begin{proposition}[Proposition~\ref{prop:Oispresheaf}]
  The above mapping $\op{\nat} \rTo^{\obs} \set$ is a presheaf.
\end{proposition}
\begin{proof}
	We need to show that $\obs$ is a functor:
	\begin{itemize}
		\item $\obs$ preserves identities: Let $\sigma\in\obs(n)$, then
		$$\sigma\cdot n=\begin{cases}\eseq&\text{if\ }n=0\wedge\sigma=\bar\tau\\\bar\tau&\text{if\ }n>0\wedge\sigma=\bar\tau\\\sigma\mid_n&\text{if\ }\sigma\neq\bar\tau\end{cases}=\begin{cases}\bar\tau&\text{if\ }n>0\wedge\sigma=\bar\tau\\\sigma&\text{if\ }\sigma\neq\bar\tau\end{cases}=\sigma$$
		\item $\obs$ respects composition: Let $\sigma\in\obs(n)$; $l\leq m\leq n$, then:
		\begin{equation*}\begin{aligned}
			&(\sigma\cdot m)\cdot l=\begin{cases}\eseq\cdot l&\text{if\ }m=0\wedge\sigma=\bar\tau\\\bar\tau\cdot l&\text{if\ }m>0\wedge\sigma=\bar\tau\\\sigma\mid_m\cdot l&\text{if\ }\sigma\neq\bar\tau\end{cases}\\
			=&\begin{cases}\eseq\mid_l&\text{if\ }m=0\wedge\sigma=\bar\tau\\\bar\tau&\text{if\ }l>0\wedge\sigma=\bar\tau\\\eseq&\text{if\ }l=0\wedge m>0\wedge\sigma=\bar\tau\\(\sigma\mid_m)\mid_l&\text{if\ }\sigma\neq\bar\tau\end{cases}=\begin{cases}\eseq&\text{if\ }l=0\wedge\sigma=\bar\tau\\\bar\tau&\text{if\ }l>0\wedge\sigma=\bar\tau\\\sigma\mid_l&\text{if\ }\sigma\neq\bar\tau\end{cases}=\sigma\cdot l
		\end{aligned}\end{equation*}
	\end{itemize}
\end{proof}

\begin{proposition}[Proposition~\ref{prop:barfxcfunctor}]
  The mapping $\bar F_X$ induced by a transition system $\tuple{X,\tact,\rightarrow}$ is a contravariant functor.
\end{proposition}
\begin{proof}
	\begin{itemize}
		\item Preservation of identities: %Identities are the trivial restrictions $\mid_{\sigma}^\sigma$ and we compute for some
Let $p\in \bar F_X(\sigma)$. Then $$p\cdot\sigma=\begin{cases}p&\text{if\ }\exists n\in\mathbb N,\sigma=(n,\overline\tau)\\p\mid_{\downarrow\sigma}&\text{if\ }\sigma\in \ftact\end{cases}=p.$$
		\item Respecting composition: Let $\sigma_3 \preceq\sigma_2 \preceq\sigma_1$ and let $p\in \bar F_X(\sigma_1)$, for some $p\in\exec X$ and $\sigma_1,\sigma_2,\sigma_3\in\ftbtact$. Then,  we need to prove
\begin{equation}\label{eq:barfunctor-comp}
  (p\cdot \sigma_2) \cdot \sigma_3 = p\cdot \sigma_3.
\end{equation}
We do the following case distinction to prove the above equation.
\begin{enumerate}
  \item Let $\sigma_1\in\ftact$. Then, by the construction of $\ftbtact$, we know that $\sigma_3,\sigma_2\in\ftact$.
      Moreover, $(p\cdot\sigma_2)\cdot\sigma_3 = (p|_{\history{\sigma_2}})|_{\history{\sigma_3}}= p|_{\history{\sigma_3}} = p\cdot\sigma_3$ since $\history{\sigma_3} \subseteq\history{\sigma_2}$.
  \item Let $\sigma_1\not\in\ftact$. Then, $\sigma_1=(n_1,\bar \tau)$, for some $n_1>0$.
      \begin{enumerate}
        \item Let $\sigma_2\in\ftact$. Then, $\sigma_2=\eseq$ and $\sigma_3=\eseq$. Moreover, $(p\cdot\sigma_2)\cdot\sigma_3 = (p|_{\{\eseq\}})|_{\{\eseq\}}=p|_{\{\eseq\}}=p\cdot\sigma_3$.
        \item Let $\sigma_2\not\in\ftact$. Then, $\sigma_2=(n_2,\bar\tau)$, for some $n_2>0$ and $n_2\leq n_1$. We further distinguish the following cases:
            \begin{enumerate}
              \item Let $\sigma_3\in\ftact$. Then, $\sigma_3=\eseq$. Moreover, $(p\cdot\sigma_2)\cdot\sigma_3=p|_{\{\eseq\}} = p\cdot\sigma_3$.
              \item Let $\sigma_3\not\in\ftact$. Then, $\sigma_3=(n_3,\bar\tau)$, for some $n_3>0$ and $n_3\leq n_2$. Moreover,
                  $(p\cdot\sigma_2)\cdot\sigma_3 = p = p\cdot\sigma_3$.
            \end{enumerate}
      \end{enumerate}
\end{enumerate}
	\end{itemize}
\end{proof}

\begin{theorem}[Theorem~\ref{thm:adjoint-barh}]
  For a given transition system $\tuple{X,\tact,\rightarrow}$, we find that
  \[
  \Sigma_{\bar h}\bar F_X(\bar\tau) \cong \coprod_{n\in\nat} \exec{X,\tau^n},
  \qquad
  \Sigma_{\bar h}\bar F_X(\varrho) \cong \mexec{X,\varrho} \quad (\text{for any $\varrho\neq\bar\tau$}).
  \]
  Moreover, the above isomorphisms are natural in $\varrho\in\ftbact$. I.e., for any $\varrho'\preceq\varrho$ with $\varrho,\varrho'\in\fseq\act$, the square in \eqref{eq:naturality-mexec} and the following square commute.
  \begin{diagram}[width=2.2cm]
    \Sigma_{\bar h} F_{X} (\bar\tau) & \rTo~{\cong} & \coprod_{n\in\nat} \exec{X,\tau^n}\\
    \dTo^{\Sigma_{\bar h}(\bar\tau,\eseq)} & & \dTo_{\_|_{\eseq}}\\
    \Sigma_h F_{X} (\eseq) & \rTo~{\cong} & \mexec{X,\eseq}
  \end{diagram}
\end{theorem}
\begin{proof}
  Most of the proof remains similar to the proof of Theorem~\ref{thm:adjoint-h}. For instance, the case when $\varrho\in\fseq\act\setminus\{\eseq\}$ remains unchanged due to the way $\bar h$ is constructed. The calculation when $\varrho=\eseq$ also remains similar except that we have to consider terms $\bar F_X(n,\bar\tau)$ (for $n>0$), which are pairwise equivalent by the way, i.e., for all $n,m>0$  we have $\bar F_X(n,\bar\tau)=F_X(m,\bar\tau)$. Moreover, the structure $\op{\ftbtact}$ is directed, i.e., we can compute the colimit using \eqref{eq:filtered-colimit}. Also, since $\bar F_X(n,\bar\tau)=\coprod_{m\in\nat} \exec{X,\tau^m}$, we have $\Sigma_{\bar h}\bar F_X(\eseq)\cong \mexec{X,\eseq}$.

  The only difference in computation is when $\varrho=\bar\tau$. In this case, we have to take the colimit over the decreasing sequence $(1,\bar\tau) \succ (2,\bar\tau) \succ \cdots$ and we observe that the action of $\bar F_X$ between any two elements in this sequence is an identity on the set $\coprod_{n\in\nat} \exec{X,\tau^n}$. Thus, we conclude that $\Sigma_{\bar h}\bar F_X(\bar\tau)\cong \coprod_{n\in\nat} \exec{X,\tau^n}$.

  Next, we need to show that the square drawn in the theorem statement commutes. For this, we recall that $\Sigma_{\bar h}\bar F_X(\bar\tau) \cong \colimit{\bar\tau \succeq \bar h(n,\bar\tau)} \bar F_X (n,\bar\tau)$. So due to the universal property of the colimits, we have the following commutative square.
  \begin{diagram}[width=2.5cm]
  \bar F_X(n,\bar\tau) & \rTo & \bar F_X(\eseq)\\
  \dTo & & \dTo\\
  \colimit{\bar\tau \succeq h(n,\bar\tau)} \bar F_X (n,\bar\tau) & \rExistMap^{\Sigma_{\bar h}(\bar\tau,\eseq)} &\colimit{\sigma\in\ftbtact} \bar F_X \sigma
 \end{diagram}
 Notice that the vertical arrows, which are the injections of the corresponding colimits, are identities. The left vertical arrow is an identity because the restriction of $\bar F_X$ between any two elements having the constant $\bar \tau $ (say, $(n',\bar\tau)\preceq (n,\bar\tau)$) is an identity, i.e., $p\cdot (n',\bar\tau)=p$ (for any $p\in \bar F_X(n,\bar\tau)$). The right vertical arrow is an identity because the elements in $\bar F_X(\eseq)$ are already minimal elements w.r.t. $\varrho=\eseq$.
 Thus, for any $p\in\bar F_X (n,\bar\tau)$, we conclude that $\Sigma_{\bar h} (\bar\tau,\eseq) (p) = p|_{\eseq}$.
\end{proof}

\begin{lemma}[Lemma~\ref{lem:semanticsfunctor}]
  The mapping $\semantics{f}$ induced by a branching simulation $f$ between $\tuple{X,\tact,\rightarrow},\tuple{Y,\tact,\rightarrow}$ is a faithful functor.
\end{lemma}
\begin{proof}
  The proof is the same as Theorem~\ref{thm:lts-tau_semantic1}, except that we have to prove the commutativity of the following naturality square when $\varrho'=\eseq\preceq \bar\tau=\varrho$.
  \begin{diagram}[width=2cm]
          \coprod_{n\in\nat} \exec{X,\tau^n} & \rTo_{\semantics{f}_{\bar\tau}} & \coprod_{n\in\nat} \exec{Y,\tau^n}\\
          \dTo^{\_|_{\eseq}} & & \dTo_{\_|_{\eseq}}\\
          \mexec{X,\eseq} & \rTo^{\semantics{f}_{\eseq}} & \mexec{Y,\eseq}
    \end{diagram}
  Let $p\in \coprod_{n\in\nat} \exec{X,\tau^n}$. Then there is some $x\in X$ such that $\eseq_x \steps{\tau^n} p$, for some $n\in\nat$. Moreover, by the definition of ${\_}_f$, we can deduce $\eseq_{f(x)} \steps{\tau^m} p_f$, for some $m\leq n$. Clearly, we have $\semantics f_\eseq (p|_{\eseq}) = \eseq_{f(x)} = (p_f)|_{\eseq}=(\semantics{f}_{\bar\tau} p)|_{\eseq }$.
  %In this case, we find that $p'\step{\tau^n}p$ and thus $p'=\epsilon_x$. Then we can directly compute $(p\cdot\sigma_{\varrho'})_f=(\epsilon_x)_f=\epsilon_{f(x)}$. On the other hand, for $p_f$ we can deduce from the definition of $\__f$, points three and four, that it has the form $\epsilon_{f(x)}\step{\tau^m}y$ for some $m\leq n$. Consequently, $p_f\cdot\sigma'_{\varrho'}=\epsilon_{f(x)}$ as expected.\qed
\end{proof}

\begin{theorem}[Theorem~\ref{thm:bbisim-char}]
  A branching simulation function $f$ is a branching bisimulation function if, and only if, the induced map $\semantics f$ is a bisimulation map in $\presheaves{\ftbact}$.
\end{theorem}
\begin{proof}
  \fbox{$\Rightarrow$} Consider the commutative square as depicted below in $\presheaves{\ftbact}$.
    \begin{equation}\label{eq:bbisim}
    \begin{diagram}
          Q & \rTo^n & \semantics{Y}\\
          \uMono^g & & \uTo_{\semantics{f}}\\
          P & \rTo_m & \semantics{X}
    \end{diagram}
    \end{equation}
    We prove the existence of $Q \rTo^k \semantics{X}$ by structural induction on the elements $\varrho\in\ftbact$. For the base case, let $\varrho=\eseq$. Notice that the map $\semantics{f}_{\eseq}$ is surjective since the given $f$ is surjective. Moreover, since every epi splits in \set, there is a map $\semantics{Y}(\eseq) \rTo^h \semantics{X}(\eseq)$ such that $\semantics f_\eseq\circ h=\id{\semantics{Y}(\eseq)}$. Thus, we define
    \[
    k_\eseq(q)=
    \begin{cases}
      m_\eseq (\inv{g}_\eseq (q)), & \text{if}\  q\in g_\eseq P(\eseq)\\
      h (n_\eseq (q)), & \text{otherwise}.
    \end{cases}
    \]
    Clearly, $\semantics{f}_\eseq\circ k_\eseq=n_\eseq$ and $k_\eseq\circ g_\eseq=m_\eseq$.

Before we tackle the inductive case, we prove the following claim:
\begin{equation}\label{eq:thm:bbisim-char-claim}
  \forall_{p\in\exec X,q\in\exec Y}\ p_f \step \tau q \implies \exists_{p'}\  (p \steps{\eseq} p' \land p'_f=q).
\end{equation}
To see this, let $p_f \step \tau q$. Then, we find that $\last{p_f} \step \tau \last q$. Using the reflection of $\tau$-transition (i.e., \eqref{eq:trans-breflect}) we find that there are states $x,x'\in X$ such that $\last {p_f} \steps{\eseq} x \step \tau x'$ and $fx=\last {p_f},fx'=\last q$. Let $\bar p,p'$ be the extension of $p$ that ends in states $x,x'$, respectively. To complete the proof, it suffices to show that $\bar p_f=p_f$. And this follows directly from the fact that $f$ satisfies the stuttering of $\tau$-transitions (i.e., \eqref{eq:b-stutter}). This completes the proof of Claim~\ref{eq:thm:bbisim-char-claim}.

For the inductive case, consider a family of functions $k_{\varrho'}$ (for each $\varrho'\preceq\varrho$) such that
\begin{enumerate}
  \item $k$ is a natural transformation up to $\varrho$, i.e., $n_{\varrho''} (q\cdot \varrho'')= (n_{\varrho'} (q)) \cdot \varrho''$, for any $\varrho''\preceq\varrho'\preceq \varrho$ and $q\in Q(\varrho')$;
  \item the lower triangle commutes up to $\varrho$, i.e., $k_{\varrho'}\circ g_{\varrho'} = m_{\varrho'}$ (for any $\varrho'\preceq\varrho$);
  \item the upper triangle commutes up to $\varrho$, i.e.,
      $\semantics{f}_{\varrho'} \circ k_{\varrho'} = n_{\varrho'}$ (for any $\varrho'\preceq \varrho$).
\end{enumerate}
We are now going to construct a function $k_{\varrho a}$ such that the desired properties (as described in the above inductive hypothesis) are satisfied.

Let $\bar q\in Q(\varrho a)$ and assume that $\bar q\not\in g_{\varrho a} P(\varrho a)$ (otherwise, simply take $k_{\varrho a}(\bar q)=m_{\varrho a}(\inv{g}_{\varrho a} (\bar q))$). Then there are some $q'\in Q(\varrho)$ and $q\in \semantics{Y}(\varrho)$ such that $q'=\bar q\cdot \varrho$, $n_\varrho (q')=q$, and $q=(n_{\varrho a} (\bar q))\cdot \varrho$. By induction hypothesis, we find a minimal execution $p\in\mexec{X,\varrho}$  such that $k_\varrho (q')=p$ and $\semantics{f}_\varrho (p)=q$. Thus, $\last{p_f} \steps{} y \step a \last{n_{\varrho a}(\bar q)}$, for some $y\in Y$. Using \eqref{eq:trans-breflect}, we find a matching sequence of transitions $\last{p} \steps{} x \step a x'$ (for some $x,x'$) such that $f(x)=y$ and $f(x')=\last{n_{\varrho a}(\bar q)}$. Let $p'$ be the extension of $p$ which witnesses the above sequence transitions that end in states $x'$. Clearly, $\bar p\in\mexec{X,\varrho a}$ and $\bar p\cdot\varrho=p$. Moreover, $\semantics{f}_{\varrho a}(\bar p) = \bar p_f = n_{\varrho a}(\bar q)$ follows directly from Claim~\ref{eq:thm:bbisim-char-claim}.
Therefore, we let $k_{\varrho a} (\bar q)=p'$. Clearly, $(k_{\varrho a}(\bar q))\cdot \varrho =k_\varrho (\bar q \cdot\varrho)$. Moreover, $\semantics{f}_{\varrho a} (k_{\varrho a}(\bar q))=n_{\varrho a}(\bar q)$; thus, the upper triangle commutes up to $\varrho a$. Lastly, commutativity of the lower triangle up to $\varrho a$ follows directly from $k_\varrho$.

The argument for the remaining case $\bar\tau$ is analogous to the induction step above, where $\eseq$ plays the role of $\varrho$. Note that then $q=\eseq_{\tilde y}$ for some $\tilde y\in Y$, $p=\eseq_{\tilde x}$ for some $\tilde x\in X$. and we find the execution $p'$ to be in $\coprod_{n\in\nat}\exec{X,\tau^n}$.

  \fbox{$\Leftarrow$} We have to show two properties: $f$ is surjective and $f$ satisfies the reflection property \eqref{eq:trans-breflect}, which we will in turn differentiate according to whether a silent or non-silent action is being considered.
	\begin{itemize}
		\item First, we show that $f$ is surjective. From Theorem~\ref{thm:bisim-epi} we know that $\semantics{f}_\eseq$ is surjective. So for any $y\in Y$, there is some $x\in X$ such that $\semantics{f}_\eseq (\eseq_x) = \eseq_{f(x)}=\eseq_y$. Thus, $f(x)=y$ and so $f$ must be surjective. %So let $y\in Y$ be given arbitrarily and let $P_\epsilon=\emptyset$, $Q_\epsilon=\epsilon_Y$ and $n_\epsilon(\epsilon_Y)=\epsilon_Y$. Then, trivially $\semantics f_\epsilon\circ m_\epsilon=\emptyset=n_\epsilon\circ g_\epsilon$. Thus, we must find a $k_\epsilon\colon \{\epsilon_Y\}\rightarrow\semantics f_\epsilon$ such that $\semantics f_\epsilon\circ k_\epsilon=n_\epsilon$. Thus, there must exist an $\epsilon_x\in\semantics X_\epsilon$, $x\in X$ such that $\semantics f_\epsilon=\epsilon_y$, i.e. $f(x)=y$. So $f$ must be surjective.
		\item Let $f(x)\step a y$, $a\in A$. We let $q$ be the execution that corresponds to the transition $f(x)\step a y$. Now construct two presheaves in $\presheaves{\ftbact}$:
\begin{itemize}
  \item $P(\eseq)=\{\eseq_x\}$ and $P(\varrho)=\emptyset$ (for $\varrho\neq\eseq$).
  \item $Q(\eseq)=\{\eseq_{f(x)}\},Q(a) =\{q\},Q(\varrho)=\emptyset$ (for $\varrho\not\in\{\eseq,a\}$) with $q\cdot \eseq =\eseq_{f(x)}$.
\end{itemize}
Consider the square in \eqref{eq:bbisim} with the maps  $g_\eseq(\eseq_x)=\eseq_{f(x)}$, $m_\eseq(\eseq_x)=\eseq_x$, $n_\eseq(\eseq_{f(x)})=\eseq_{f(x)}$ and $n_a(q)=q$. Clearly, $g$ is a mono in $\presheaves{\ftbact}$ and the square in \eqref{eq:bbisim} commutes, i.e., $\semantics f\circ m = n\circ g$. Since $\semantics f$ is assumed to be a bisimulation map, there is a map $Q \rTo^k\semantics X$ such that the upper and lower triangle in \eqref{eq:bbisim} commute, i.e., $\semantics f\circ k=n$ and $k\circ g = m$.
		
		We compute $\semantics f_a\circ k_a(q)=n_a(q)=q$. Thus we find a $p\in\semantics X(a)$ such that $p_f=q$ and by restriction to $\eseq$ using $m=k\circ g$, we also obtain $p\mid_\eseq=\eseq_x$. By the definition of $p_f$, this means that $x\steps\tau x'\step a x''$ and $f(x'')=y$ as required.
		
		\item The case of a $\tau$-step $f(x)\step\tau y$ (corresponding to an execution $q$ observable at $\bar\tau$) can be shown analogously to the visible $a$-step, where we consider the set $Q(\bar\tau)=\{q\}$ instead of $Q(a)$ and argue identically, else.
	\end{itemize}
\end{proof}

%\begin{thebibliography}{8}
%\bibitem{ref_article1}
%Author, F.: Article title. Journal \textbf{2}(5), 99--110 (2016)
%
%\bibitem{ref_lncs1}
%Author, F., Author, S.: Title of a proceedings paper. In: Editor,
%F., Editor, S. (eds.) CONFERENCE 2016, LNCS, vol. 9999, pp. 1--13.
%Springer, Heidelberg (2016). \doi{10.10007/1234567890}
%
%\bibitem{ref_book1}
%Author, F., Author, S., Author, T.: Book title. 2nd edn. Publisher,
%Location (1999)
%
%\bibitem{ref_proc1}
%Author, A.-B.: Contribution title. In: 9th International Proceedings
%on Proceedings, pp. 1--2. Publisher, Location (2010)
%
%\bibitem{ref_url1}
%LNCS Homepage, \url{http://www.springer.com/lncs}. Last accessed 4
%Oct 2017
%\end{thebibliography}
\end{document}